\documentclass[a4paper,11pt]{article}
\usepackage{amsmath,amssymb,amsfonts,amsthm,amscd}
\usepackage{mathrsfs}
\usepackage{a4wide}

\newcommand{\com}[2]{[#1,#2]}
\newcommand{\pair}[2]{\langle #1,#2\rangle}

\def\vol{\mathrm{vol}}
\def\bbR{\mathbb{R}}
\newtheorem{propo}{Proposition}
\newtheorem{lem}{Lemma}
\newtheorem{cor}{Corollary}

\theoremstyle{definition}
\newtheorem{defi}{Definition}
\newtheorem{example}{Example}

\theoremstyle{remark}
\newtheorem{remark}{Remark}

\title{%
  QFT  on homothetic Killing twist deformed curved spacetimes}

\author{%
  Alexander Schenkel%
    \thanks{e-mail: \texttt{aschenkel@physik.uni-wuerzburg.de}}\\
  \hfil\\
  Institut f\"ur Theoretische Physik und Astrophysik\\
  Universit\"at W\"urzburg, Am Hubland, 97074 W\"urzburg, Germany}

\date{September 2010}
\begin{document}
\maketitle

\begin{abstract}
We study the quantum field theory (QFT) of a free, real, massless and curvature coupled scalar field on self-similar
symmetric spacetimes, which are deformed by an abelian Drinfel'd twist constructed from a Killing 
and a homothetic Killing vector field. 
In contrast to deformations solely by Killing vector fields, such as the Moyal-Weyl Minkowski spacetime, 
the equation of motion and Green's operators are deformed.
We show that there is a $\ast$-algebra isomorphism between the QFT on the deformed and
the formal power series extension of the QFT on the undeformed spacetime.
We study the convergent implementation of our deformations for toy-models. For these models it is found that there
is a $\ast$-isomorphism between the deformed Weyl algebra and a reduced undeformed Weyl algebra,
where certain strongly localized observables are excluded.
Thus, our models realize the intuitive physical picture that noncommutative geometry prevents arbitrary
localization in spacetime.

\end{abstract}
\newpage

\section{Introduction}
Replacing classical geometry by a noncommutative (NC) geometry is a widely used technique to 
incorporate quantum gravity effects into (low-energy) QFTs. 
By now, there are different realizations of the basic idea of modifying/deforming commutative QFTs 
using methods of NC geometry. In particular, the following different routes have been taken:
\begin{flalign*}
\begin{CD}
\text{spacetime} @> \quad\text{deform}\quad  >>  \text{deformed spacetime} \\
@V \text{quantum fields} VV                   				@VV \text{quantum fields} V\\
\text{QFT}   @> \quad\text{deform} \quad>> \text{deformed QFT}
\end{CD}
\end{flalign*}
The resulting deformed QFT is expected to depend on the path chosen.
The lower path, which we call NC QFT on curved spacetimes, assumes that the QFT is the fundamental object
to be studied and spacetime is just encoded in the net structure of the algebras of local observables. 
Wedge local deformations of QFTs on curved spacetimes by Killing vector fields 
and Rieffel products have been studied 
in \cite{Dappiaggi:2010hc}.
The upper path, which we call QFT on NC curved spacetimes, assumes that spacetime itself is of direct importance
and has to be deformed. The QFT is then defined on this NC spacetime. We have constructed the QFT on NC curved spacetimes
 for a free scalar field and a large class of formal Drinfel'd twist deformations in \cite{Ohl:2009qe}, without
assuming any symmetries of the background spacetime.

In the present work we specialize the formalism of \cite{Ohl:2009qe} to a certain class of deformations.
The deformations we consider are abelian Drinfel'd twists constructed from a Killing vector field and a homothetic
Killing vector field. The reason for this choice is threefold: Firstly, as pointed out in \cite{Ohl:2009qe},
the equation of motion, the Green's operators and also the algebra of observables of a free scalar QFT are not deformed
when one uses deformations by Killing vector fields. Thus, Killing deformations are trivial in the formalism \cite{Ohl:2009qe}.
As we will see later, deforming by a Killing and a homothetic Killing vector field is the simplest nontrivial deformation.
Secondly, due to the simplicity of the deformation we can study its convergent implementation for toy-models 
and the resulting new physical effects, such
as nonlocality. Thirdly, as pointed out in \cite{Aschieri:2009qh}, homothetic Killing vector fields
play an important role in constructing exact solutions of the NC Einstein equations \cite{Aschieri:2005yw, Aschieri:2005zs}.

The outline of this paper is as follows: In Section \ref{sec:homothetic} we introduce homothetic
Killing vector fields and provide examples of spacetimes admitting them. We review the formal deformation quantization
of algebras and their modules by (abelian) Drinfel'd twists in Section \ref{sec:deformation} in order
to define a deformed action functional of a free, real, massless and curvature coupled scalar field theory.
In Section \ref{sec:formal} we study in detail the scalar field theory on curved spacetimes 
deformed by an abelian twist constructed from a Killing and a homothetic Killing vector field.
We construct the $\ast$-algebra of field polynomials and show that there is a $\ast$-algebra isomorphism
to the formal power series extension of the $\ast$-algebra of field polynomials of the commutative QFT.
As a consequence, there also exists an isomorphism between the groups of symplectic automorphisms (e.g.~isometries)
and the spaces of algebraic states of the deformed and the undeformed QFT. In Section \ref{sec:convergent}
we study the convergent implementation of our deformations for Friedmann-Robertson-Walker (FRW) toy-models. 
We show that there is a $\ast$-isomorphism between the deformed Weyl algebra $A_\star$ and a nonstandard undeformed Weyl algebra
$A_\text{img}$.
The algebra $A_\text{img}$ does not carry a representation of all isometries of the FRW spacetime.
However, there exists an injective $\ast$-morphism $\mathfrak{S}$ from $A_\star$ into an (i.g.~extended) undeformed
Weyl algebra $A_\text{ext}$, which carries a representation of all isometries. 
Thus, we can induce symmetric states on $A_\star$ by pulling-back symmetric states
on $A_\text{ext}$.
In order to show that our models can, in principle, lead to observable features we briefly discuss
the NC modifications to the cosmological power spectrum of a scalar field, i.e.~of the two-point 
correlation function.
We conclude and give an outlook in Section \ref{sec:conc}.
In the Appendix \ref{app:twisted}
we show that the so-called twisted QFT construction \cite{Zahn:2006wt,Balachandran:2007vx,Aschieri:2007sq}, 
i.e.~a realization of the lower path in the diagram above, can only be done if all homothetic Killing 
vector fields entering the twist are Killing.

\section{\label{sec:homothetic}Homothetic Killing vector fields}
Let $(\mathcal{M}, g)$ be a $D$-dimensional smooth Lorentzian manifold of signature $(-,+,\cdots,+)$.
 We denote by $\vol$ the canonical metric volume form and by $\mathfrak{R}$ the curvature of $g$. The algebra of smooth, complex
 valued functions on $\mathcal{M}$ is denoted by $\mathcal{A}=\bigl(C^{\infty}(\mathcal{M}),\cdot\bigr)$, where 
 $\cdot$ is the usual point-wise multiplication. 
 We denote by $\Xi$ the smooth complex vector fields on $\mathcal{M}$, i.e.~the complexified smooth sections of the tangent bundle.
\begin{defi}
A vector field $v\in\Xi$ is called a Killing vector field, if $\mathcal{L}_v (g)=0$.
It is called a homothetic Killing vector field, if $\mathcal{L}_v(g)= c_v\, g$
with $c_v\in\mathbb{C}$. $v$ is called proper, if $c_v\neq 0$.
\end{defi}
Obviously, each Killing vector field is also a homothetic Killing vector field with $c_v=0$. 
A homothetic Killing vector field is a special case of a conformal Killing vector field $v\in\Xi$,
 satisfying $\mathcal{L}_v(g)=h\,g$ with $h\in\mathcal{A}$. We do not discuss general conformal Killing vector fields
 in this work.

We remind the reader of the following standard result.
\begin{propo}
The homothetic Killing vector fields form a Lie subalgebra $(\mathfrak{H},\com{~}{~})\subseteq (\Xi,\com{~}{~})$
of the Lie algebra of vector fields on $\mathcal{M}$.
The Killing vector fields form a Lie subalgebra $(\mathfrak{K},\com{~}{~})\subseteq (\mathfrak{H},\com{~}{~})$ and
the following inclusion holds true
\begin{flalign}
\com{\mathfrak{H}}{\mathfrak{H}}\subseteq \mathfrak{K}~.
\end{flalign}
\end{propo}
\begin{proof}
Let $v,w\in\mathfrak{H}$. Then $\mathcal{L}_{\beta\,v+\gamma\,w}(g)=\beta\mathcal{L}_v(g)+\gamma\mathcal{L}_w(g)
=(\beta\,c_v +\gamma\,c_w)g$, for all $\beta,\gamma\in\mathbb{C}$. Thus, $\mathfrak{H}$ is a vector space over $\mathbb{C}$
and $\mathfrak{K}\subseteq\mathfrak{H}$ is a vector subspace.
Furthermore,
\begin{flalign}
\mathcal{L}_{\com{v}{w}}(g)=(\mathcal{L}_{v}\circ\mathcal{L}_w-\mathcal{L}_w\circ\mathcal{L}_v)(g)= (c_v\,c_w-c_w\,c_v)g=0~,
\end{flalign}
such that $[v,w]\in\mathfrak{K}\subseteq \mathfrak{H}$. For $v,w\in\mathfrak{K}$ we trivially find $\com{v}{w}\in\mathfrak{K}$.
\end{proof}
\noindent 
It can be shown that $\text{dim}(\mathfrak{K})\leq\text{dim}(\mathfrak{H})\leq\text{dim}(\mathfrak{K})+1$.
 To prove this, assume
that there are two proper homothetic Killing vector fields $v,w\in\mathfrak{H}$, satisfying
$\mathcal{L}_v(g)=c_v\,g$ and $\mathcal{L}_w(g)=c_w\,g$ with $c_v,c_w\neq0$.
Then $u:=c_w\, v-c_v\,w$ is a Killing vector field, since
 $\mathcal{L}_u(g)=\mathcal{L}_{c_w\, v-c_v\,w}(g)=(c_w c_v-c_v c_w)g=0$, and
 $w=(c_w\,v-u)/c_v$ is a linear combination of a proper homothetic Killing vector field and a Killing vector field.

Let us provide some examples of manifolds allowing for proper homothetic Killing vector fields.
\begin{example}
\label{ex:mink}
Let $\mathcal{M}=\mathbb{R}^D$ and let $g=\eta_{\mu\nu}dx^\mu\otimes dx^\nu$,
where $x^\mu$ are global coordinates on $\mathbb{R}^D$ and $\eta_{\mu\nu}=\text{diag}(-1,1,\dots,1)_{\mu\nu}$.
It is well-known that $\mathfrak{K}$ is the Lie algebra of the Poincar{\'e} group $SO(D-1,1)\ltimes \mathbb{R}^D$.
A proper homothetic Killing vector field is given by the dilation $v=x^\mu\partial_\mu$,
satisfying $\mathcal{L}_v(g) = 2\,g$.
\end{example}
\begin{example}[\cite{Eardley:1973fm}]
\label{ex:frw}
Let $\mathcal{M}=(0,\infty) \times \mathbb{R}^{D-1}$ and let $g=-dt\otimes dt +a(t)^2\,\delta_{ij}dx^i\otimes dx^j$,
where $t\in(0,\infty)$ is the cosmological time, $x^i,~i=1,\dots,\,D-1$, are comoving coordinates and $a(t)$ is the scale
factor of the universe. If we assume that $a(t)\propto t^p$, where $p\in\mathbb{R}$, we have a
proper homothetic Killing vector field
\begin{flalign}
v=t\partial_t +\left(1-t\,\frac{\dot a(t)}{a(t)}\right)x^i\partial_i=t\partial_t+(1-p)\,x^i\partial_i~,
\end{flalign}
satisfying $\mathcal{L}_v(g)= 2\,g$. These spacetimes are relevant in cosmology, since a perfect fluid with
equation of state $P=\omega\rho$, where $P$ is the pressure, $\rho$ is the energy density and $\omega\in\mathbb{R}$
is a parameter, leads to a scale factor $a(t)\propto t^{\frac{2}{3(\omega+1)}}$, i.e.~$p=\frac{2}{3(\omega+1)}$. 
\end{example}
For more examples, including the Kasner spacetime and the  plane-wave spacetime, as well as a 
construction principle for spacetimes allowing for a proper homothetic Killing vector field see \cite{Eardley:1973fm}.


\section{\label{sec:deformation}Formal deformation quantization by abelian Drinfel'd twists}
We follow the approach of \cite{Aschieri:2005yw, Aschieri:2005zs}, but restrict ourselves to a relatively simple class
of deformations, the so-called abelian twists \cite{Gerstenhaber,Reshetikhin:1990ep,Jambor:2004kc}
\begin{flalign}
\mathcal{F}=\exp\left(-\frac{i\lambda}{2}\Theta^{ab}X_a\otimes X_b\right)~.
\end{flalign}
Here $X_a\in\Xi$ are mutually commuting vector fields, i.e.~$\com{X_a}{X_b}=0~\forall_{ab}$, 
$\Theta^{ab}$ is real, constant and antisymmetric and $\lambda$ is the deformation parameter. 
We assume the $X_a$ to be real.
The inverse twist is given by (sum over $\alpha$ understood)
\begin{flalign}
\label{eqn:twist}
\mathcal{F}^{-1}=\bar f^\alpha\otimes \bar f_\alpha= \exp\left(\frac{i\lambda}{2}\Theta^{ab}X_a\otimes X_b\right)~.
\end{flalign}
The commutative algebra of functions $\mathcal{A}=\bigl(C^\infty(\mathcal{M}),\cdot\bigr)$
 is deformed into a NC associative algebra by introducing the $\star$-product
 \begin{flalign}
 h\star k := \bar f^\alpha( h )\cdot \bar f_\alpha( k )~,
 \end{flalign}
 for all $h,k\in C^\infty(\mathcal{M})$.
 The vector fields act on functions via the Lie derivative.
The precise definition of the deformed algebra is $\mathcal{A}_\star:= \bigl(C^\infty(\mathcal{M})[[\lambda]],\star\bigr)$,
where $[[\lambda]]$ denotes formal power series. 

Similarly, we deform the differential calculus $(\Omega^\bullet,\wedge,d)$ of differential forms on $\mathcal{M}$ by
using the twist (\ref{eqn:twist}). We define the deformed wedge product by
\begin{flalign}
\omega\wedge_\star \omega^\prime := \bar f^\alpha( \omega )\wedge  \bar f_\alpha ( \omega^\prime )~,
\end{flalign}
for all $\omega, \omega^\prime\in \Omega^\bullet$. The vector fields act via the Lie derivative on differential forms.
It turns out that the exterior differential $d$ satisfies the Leibniz rule
\begin{flalign}
d(\omega\wedge_\star \omega^\prime) = (d\omega)\wedge_\star \omega^\prime + (-1)^{\text{deg}(\omega)} \omega\wedge_\star(d\omega^\prime)~,
\end{flalign}
for all $\omega,\omega^\prime\in\Omega^\bullet$, since Lie derivatives commute with $d$. 
Thus, the exterior differential does not have to be deformed. 
The deformed differential calculus is given by $\Omega^\bullet_\star:=(\Omega^\bullet[[\lambda]],\wedge_\star,d)$.
Note that the formal power series of $n$-forms $\Omega_\star^n:=\Omega^{n}[[\lambda]]$ are $\mathcal{A}_\star$-bimodules,
where the left and right $\mathcal{A}_\star$-action is provided by the deformed  wedge product.

We deform the vector fields $\Xi$ into an $\mathcal{A}_\star$-bimodule of deformed vector fields $\Xi_\star := \Xi[[\lambda]]$ by 
employing the deformed left and right $\mathcal{A}_\star$-action
\begin{flalign}
h\star v:=\bar f^\alpha(h)\,\bar f_\alpha(v)~,\quad v\star h:=\bar f^\alpha(v)\,\bar f_\alpha(h)~,
\end{flalign}
for all $h\in\mathcal{A}_\star$ and $v\in\Xi_\star$. 
The action of the twist  on $\Xi[[\lambda]]$ is given by the Lie derivative.
The duality pairing between vector fields and one-forms can also be deformed by defining
\begin{flalign}
\pair{v}{\omega}_\star:= \pair{\bar f^\alpha(v)}{\bar f_\alpha(\omega)}~,
\end{flalign}
for all $v\in\Xi_\star$ and $\omega\in\Omega^1_\star$. The $\star$-pairing 
with $v$ on the right and $\omega$ on the left is defined analogously. One obtains the following relevant property
by using identities of the twist
\begin{flalign}
\label{eqn:pairinglinearity}
\pair{h\star v\star k}{\omega\star l}_\star=h\star\pair{v}{k \star\omega}_\star \star l~,
\end{flalign}
for all $h,k,l\in\mathcal{A}_\star$, $v\in\Xi_\star$ and $\omega\in \Omega^1_\star$.

Employing the $\star$-tensor product
\begin{flalign}
\tau\otimes_\star\tau^\prime := \bar f^\alpha(\tau)\otimes\bar f_\alpha(\tau^\prime)~,
\end{flalign}
we can deform the tensor algebra $(\mathcal{T},\otimes)$ over $\mathcal{A}$
generated by $\Xi$ and $\Omega^1$ into the tensor algebra $(\mathcal{T}_\star,\otimes_\star)$ 
over $\mathcal{A}_\star$ generated by $\Xi_\star$ and $\Omega^1_\star$.
The relation (\ref{eqn:pairinglinearity}) extends to the $\star$-tensor algebra 
\begin{flalign}
\pair{\tau\otimes_\star v\star h}{\omega\otimes_\star\tau^\prime}_\star= \tau\otimes_\star \pair{v}{h\star \omega}_\star \star\tau^\prime~,
\end{flalign}
for all $h\in\mathcal{A}_\star$, $v\in\Xi_\star$, $\omega\in \Omega^1_\star$ and $\tau,\tau^\prime\in\mathcal{T}_\star$.

Using these tools we are in the position to construct a deformed action functional for a free, real, massless 
and curvature coupled scalar field $\Phi$.
We follow \cite{Ohl:2009qe} and define
\begin{flalign}
\label{eqn:action}
 S_\star :=-\frac{1}{2}\int\limits_\mathcal{M}\left(\pair{\pair{d\Phi}{g^{-1_\star}}_\star}{d\Phi}_\star + \xi \,\Phi\star \mathfrak{R}\star \Phi\right) \star \text{vol}~,
\end{flalign}
where $\xi\in\mathbb{R}$ and $g^{-1_\star}\in\Xi_\star\otimes_\star\Xi_\star$ 
is the $\star$-inverse metric of $g$ defined by $\pair{g^{-1_\star}}{\pair{g}{v}_\star}_\star=v$
and $\pair{g}{\pair{g^{-1_\star}}{\omega}_\star}_\star=\omega$, for all $v\in\Xi_\star$ and $\omega\in\Omega^1_\star$.


\section{\label{sec:formal}Free, real, massless and curvature coupled scalar field on homothetic Killing deformed spacetimes}
\subsection{Deformed equation of motion}
We study in detail a subclass of the deformations (\ref{eqn:twist}), where 
$a,b\in\lbrace 1,2 \rbrace$, $X_1\in\mathfrak{K}$ and $X_2\in\mathfrak{H}$.
We use the normalization $\Theta^{12}=-\Theta^{21}=1$, $\Theta^{11}=\Theta^{22}=0$
and $\mathcal{L}_{X_2}(g) = c \,g$, $c\in\mathbb{R}$. We call these
deformations homothetic Killing deformations by two vector fields.
\begin{propo}
 Consider a homothetic Killing deformation by two vector fields $X_1\in\mathfrak{K}$ and $X_2\in\mathfrak{H}$
of a $D$-dimensional smooth Lorentzian manifold $(\mathcal{M},g)$. 
Then $g^{-1_\star} = g^{-1}$, where $g^{-1}\in\Xi\otimes\Xi$ is the undeformed inverse metric field.
The equation of motion corresponding to the scalar field action (\ref{eqn:action}) is given by
(suppressing the symbol $\mathcal{L}$ for Lie derivatives)
\begin{flalign}
\label{eqn:eomoperator}
 \widehat{P}_\star(\Phi)= \cos\left(\frac{\lambda c}{4}\left(D+2\right)X_1\right)\,\left(\square_g -\xi\, \mathfrak{R}\right)\Phi=0~,
\end{flalign}
where $\square_g$ is the undeformed d'Alembert operator.
\end{propo}
\begin{proof}
The $\star$-inverse metric $g^{-1_\star}\in\Xi_\star\otimes_\star\Xi_\star$ of $g$ exists and is unique. 
We show that the undeformed inverse metric $g^{-1}\in\Xi\otimes\Xi$ defined by $\pair{g^{-1}}{\pair{g}{v}}=v$ 
and $\pair{g}{\pair{g^{-1}}{\omega}}=\omega$, for all $v\in\Xi$ and $\omega\in\Omega^1$, is equal to $g^{-1_\star}$.
For $g^{-1}$ one easily finds $\mathcal{L}_{X_1}(g^{-1})=0$ and $\mathcal{L}_{X_2}(g^{-1})=-c\,g^{-1}$.
Using the homothetic Killing property we obtain
\begin{flalign}
\label{eqn:hkprop1}
\pair{g}{v}_\star = \pair{g}{e^{-\frac{i\lambda c}{2}X_1}v}~,\quad \pair{g^{-1}}{\omega}_\star=\pair{g^{-1}}{e^{\frac{i\lambda c}{2}X_1}\omega}~,
\end{flalign}
for all $v\in\Xi_\star$ and $\omega\in\Omega^1_\star$. Thus,
\begin{flalign}
 \pair{g^{-1}}{\pair{g}{v}_\star}_\star = \pair{g^{-1}}{\pair{g}{v}}=v~,\quad
\pair{g}{\pair{g^{-1}}{\omega}_\star}_\star = \pair{g}{\pair{g^{-1}}{\omega}}=\omega~,
\end{flalign}
by using the invariance of $g$ and $g^{-1}$ under $X_1$.

\noindent For the metric volume form one finds that $\mathcal{L}_{X_1}(\vol)=0$ and $\mathcal{L}_{X_2}(\vol)=\frac{c D}{2}\vol$.
For the curvature we have $\mathcal{L}_{X_1}(\mathfrak{R})=0$ and $\mathcal{L}_{X_2}(\mathfrak{R})=-c\,\mathfrak{R}$ \cite{Aschieri:2009qh}.
Using this, (\ref{eqn:hkprop1}) and graded cyclicity in order to remove one $\star$ under the integral, we obtain
for the variation of the action (\ref{eqn:action}) by functions $\delta\Phi$ of compact support
\begin{flalign}
 \delta S_\star = 
\int\limits_\mathcal{M}\delta\Phi\,\vol\,\cos\left(\frac{\lambda c}{4}(D+2)X_1\right)\,(\square_g-\xi\,\mathfrak{R}) \Phi~.
\end{flalign}
\end{proof}
\begin{remark}
\label{remark:eomoperators}
 The scalar-valued equation of motion operator $\widehat{P}_\star$ is defined slightly different to
the one in \cite{Ohl:2009qe}. The variation of the action (\ref{eqn:action}) yields a top-form valued equation of motion
operator. In \cite{Ohl:2009qe} we have used the NC Hodge operator
 $\mathcal{A}_\star\to\Omega_\star^D,~\varphi\mapsto\varphi\star\vol$ in order to extract 
the scalar-valued equation of motion operator $P_\star$, while here we have used the classical Hodge operator
$\mathcal{A}_\star\to\Omega_\star^D,~\varphi\mapsto \varphi\,\vol$ to extract $\widehat{P}_\star$.
For the class of models discussed in the present work the relation between these two equation 
of motion operators is given by
\begin{flalign}
 P_\star = e^{-\frac{i\lambda c D}{4}X_1}\circ \widehat{P}_\star~.
\end{flalign}
\end{remark}
Note that in case we deform by two Killing vector fields $X_1,X_2\in\mathfrak{K}$ we have $c=0$ and the equation
of motion operator $\widehat{P}_\star$ (\ref{eqn:eomoperator}) is undeformed. This is a generalization of the well-known
result that the dynamics of a free scalar field theory on the Moyal-Weyl deformed Minkowski spacetime is undeformed.

\subsection{Deformed Green's operators and the symplectic $\bbR[[\lambda]]$-module}
Let $(\mathcal{M},g)$ be a connected, time-oriented and globally hyperbolic smooth Lorentzian manifold.
The construction of the Green's operators corresponding to the deformed equation of motion operator (\ref{eqn:eomoperator})
is straightforward. 
We define
\begin{flalign}
\label{eqn:greenoperator}
 \widehat{\Delta}_{\star\pm}:= \Delta_{\pm}\circ \cos\left(\frac{\lambda c}{4}\left(D+2\right)X_1\right)^{-1}~,
\end{flalign}
where the inverse of $\cos\left(\frac{\lambda c}{4}\left(D+2\right)X_1\right)$ is understood
in terms of formal power series and $\Delta_\pm$ are the unique retarded and advanced
 Green's operators corresponding to the undeformed equation of motion operator $P=\square_g-\xi\,\mathfrak{R}$ \cite{baer}. We find
\begin{subequations}
\begin{flalign}
 \widehat{P}_\star\circ \widehat{\Delta}_{\star\pm} & = 
\text{id}\vert_{C^\infty_0(\mathcal{M})[[\lambda]]}~,\\
\widehat{\Delta}_{\star\pm}\circ \widehat{P}_\star\vert_{C^\infty_0(\mathcal{M})[[\lambda]]} 
&= \text{id}\vert_{C^\infty_0(\mathcal{M})[[\lambda]]}~,
\end{flalign}
and the support property
\begin{flalign}
 \text{supp}\bigl(\widehat{\Delta}_{\star\pm}(\varphi)\bigr)\subseteq J^\pm\bigl(\text{supp}(\varphi)\bigr)~,
\end{flalign}
\end{subequations}
for all $\varphi\in C_0^\infty(\mathcal{M})$, since the NC corrections to $\Delta_{\pm}$ are finite order differential operators
at every order in $\lambda$. 
Here $J^\pm(A)$ denotes the causal future/past of a subset $A\subseteq\mathcal{M}$ with respect to the metric field $g$.
Note that the support property relies on the use of formal power series.
Convergent deformations in general {\it do not} preserve classical causality, as it is shown in Section \ref{sec:convergent}
for explicit examples.
\begin{remark}
 In \cite{Ohl:2009qe} we have studied the Green's operators $\Delta_{\star\pm}$ of the equation of motion operator 
$P_\star$, see Remark \ref{remark:eomoperators}. For the class of models discussed in the present work
 the relation to the Green's operators $\widehat{\Delta}_{\star\pm}$ is given by
\begin{flalign}
 \Delta_{\star\pm}=\widehat{\Delta}_{\star\pm}\circ e^{\frac{i\lambda c D}{4}X_1}~.
\end{flalign}
\end{remark}

Following \cite{Ohl:2009qe} we can construct a deformed symplectic $\bbR[[\lambda]]$-module\footnote{
In deformation quantization the fields $\bbR$ and $\mathbb{C}$ are replaced by the
commutative and unital rings $\bbR[[\lambda]]$ and $\mathbb{C}[[\lambda]]$. 
Similarly, vector spaces over $\bbR$ and $\mathbb{C}$ are replaced by modules
over $\bbR[[\lambda]]$ and $\mathbb{C}[[\lambda]]$. 
A (weak) symplectic $\bbR[[\lambda]]$-module $(W,\rho)$ is an $\bbR[[\lambda]]$-module
$W$ with an antisymmetric and $\bbR[[\lambda]]$-bilinear map $\rho:W\times W\to\bbR[[\lambda]]$, such that
$\rho(\varphi,\psi)=0$ for all $\psi\in W$ implies $\varphi =0$. Following \cite{baer} we suppress the term {\it weak}.
In \cite{Ohl:2009qe} the symplectic $\bbR[[\lambda]]$-module was (with an abuse of notation) referred to 
symplectic vector space over $\bbR[[\lambda]]$, however in this work we intend to be more precise in notation.
}, which is isomorphic 
to the space of solutions of $P_\star$.
In terms of the unhatted quantities $\Delta_{\star\pm}$ and $P_\star$ the pre-symplectic $\bbR[[\lambda]]$-module is given by
$\bigl(H,\omega_\star\bigr)$, where 
\begin{flalign}
 H:=\bigl\lbrace\varphi\in C^\infty_0(\mathcal{M})[[\lambda]]:
\left(\Delta_{\star\pm}(\varphi)\right)^\ast=\Delta_{\star\pm}(\varphi)\bigr\rbrace
\end{flalign}
is the space of physical sources and
\begin{flalign}
 \omega_\star:H\times H\to\mathbb{R}[[\lambda]],~(\varphi,\psi)\mapsto \omega_\star(\varphi,\psi)
=\int\limits_\mathcal{M} \varphi^\ast\star\Delta_\star(\psi)\star\vol 
\end{flalign}
is an $\mathbb{R}[[\lambda]]$-bilinear and antisymmetric map. The deformed fundamental solution is defined
by $\Delta_\star:= \Delta_{\star+}-\Delta_{\star-}$.
\begin{propo}
\label{propo:hattedunhattedsymplecto}
For homothetic Killing deformations by two vector fields $X_1\in\mathfrak{K}$ and $X_2\in\mathfrak{H}$ we have
\begin{flalign}
 H=\left\lbrace e^{-\frac{i\lambda c D}{4}X_1} \varphi:\varphi\in C^\infty_0(\mathcal{M},\mathbb{R})[[\lambda]]\right\rbrace
\end{flalign}
and
\begin{flalign}
 \widehat{\omega}_\star(\varphi,\psi):=\omega_\star\left(e^{-\frac{i\lambda c D}{4}X_1} \varphi,e^{-\frac{i\lambda c D}{4}X_1} \psi\right)
=\int\limits_\mathcal{M} \varphi\,\widehat{\Delta}_{\star}(\psi)\,\vol~,
\end{flalign}
for all $\varphi,\psi\in C^\infty_0(\mathcal{M},\mathbb{R})[[\lambda]]$, where 
$\widehat{\Delta}_\star:=\widehat{\Delta}_{\star+}-\widehat{\Delta}_{\star-}$.

\noindent The map $C^\infty_0(\mathcal{M},\mathbb{R})[[\lambda]]\to H,~
\varphi\mapsto e^{-\frac{i\lambda c D}{4}X_1} \varphi$ is a symplectic isomorphism
between the pre-symplectic $\bbR[[\lambda]]$-modules 
$\bigl(C^\infty_0(\mathcal{M},\mathbb{R})[[\lambda]],\widehat{\omega}_\star\bigr)$
and $\bigl(H,\omega_\star\bigr)$. 
\end{propo}
\begin{proof}
 Let $\varphi\in C^\infty_0(\mathcal{M},\mathbb{R})[[\lambda]]$, then $e^{-\frac{i\lambda c D}{4}X_1}\varphi\in H$, since
\begin{flalign}
 \left(\Delta_{\star\pm}\left( e^{-\frac{i\lambda c D}{4}X_1} \varphi\right)\right)^\ast = 
\left(\widehat{\Delta}_{\star\pm}(\varphi)\right)^\ast
=\widehat{\Delta}_{\star\pm}\left(\varphi^\ast\right)=\widehat{\Delta}_{\star\pm}(\varphi)
=\Delta_{\star\pm}\left( e^{-\frac{i\lambda c D}{4}X_1}\varphi\right)~.
\end{flalign}
Let now $\varphi\in H$, then there is a $\psi\in C^\infty_0(\mathcal{M})[[\lambda]]$, such that 
$\varphi=e^{-\frac{i\lambda c D}{4}X_1} \psi$ (due to the invertability of $e^{-\frac{i\lambda c D}{4}X_1}$).
It holds $\psi^\ast=\psi$, since
\begin{flalign}
 \left(\Delta_{\star\pm}(\varphi)\right)^\ast=\Delta_{\star\pm}(\varphi)\quad \Leftrightarrow \quad
 \widehat{\Delta}_{\star\pm}\left(\psi^\ast-\psi\right)= 0\quad \Leftrightarrow\quad \psi^\ast-\psi=0~. 
\end{flalign}
Let $\varphi,\psi\in C^\infty_0(\mathcal{M},\mathbb{R})[[\lambda]]$, then
\begin{multline}
 \widehat{\omega}_\star(\varphi,\psi):= \omega_\star\left(e^{-\frac{i\lambda c D}{4}X_1}\varphi,e^{-\frac{i\lambda c D}{4}X_1}\psi\right)=
\int\limits_\mathcal{M}\left(e^{-\frac{i\lambda c D}{4}X_1}\varphi\right)^\ast\star\widehat{\Delta}_{\star}(\psi)\star\vol\\
~~\stackrel{\text{GC, REvol}}{=}~~\int\limits_\mathcal{M} \left(e^{-\frac{i\lambda c D}{4}X_1}\varphi\star \vol\right)^\ast
\,\widehat{\Delta}_\star(\psi)~~\stackrel{\text{HKP, REvol, RE}\varphi}{=}~~\int\limits_\mathcal{M}\varphi\,\widehat{\Delta}_\star(\psi)\,\vol~.
\end{multline}
In this derivation we have used graded cyclicity (GC), the reality of $\vol$ (REvol), the homothetic Killing property
$\mathcal{L}_{X_2}(\vol)=\frac{c D}{2}\,\vol$ (HKP) and the reality of $\varphi$ (RE$\varphi$).
Graded cyclicity allows us to perform (graded) cyclic permutations and remove one of the $\star$-products under the integral.
\end{proof}
The pre-symplectic $\bbR[[\lambda]]$-modules $\bigl(H,\omega_\star\bigr)$ and 
$\bigl(C^\infty_0(\mathcal{M},\mathbb{R})[[\lambda]],\widehat{\omega}_\star\bigr)$ can be made symplectic by factoring out
 the kernel of $\Delta_{\star}$ and $\widehat{\Delta}_\star$, respectively.
The kernel of $\Delta_{\star}$ is given by $P_\star[C^\infty_0(\mathcal{M},\mathbb{R})[[\lambda]]]$ and the one of
$\widehat{\Delta}_\star$ by $\widehat{P}_\star[C^\infty_0(\mathcal{M},\mathbb{R})[[\lambda]]]$.
Note that these kernels are related by the symplectic isomorphism of Proposition \ref{propo:hattedunhattedsymplecto}.
\begin{remark}
Let us make some comments on the physical interpretation of the pre-symplectic $\bbR[[\lambda]]$-modules
$\bigl(H,\omega_\star\bigr)$ and
$\bigl(C^\infty_0(\mathcal{M},\mathbb{R})[[\lambda]],\widehat{\omega}_\star\bigr)$.
Later in QFT the elements of $\bigl(H,\omega_\star\bigr)$ are interpreted as smearing functions of 
the operator valued distributions $\hat\Phi(x)$ (the hermitian field operators), where the smearing is given by
$ \hat\Phi(\varphi)=\int_\mathcal{M}\hat\Phi\star \varphi\star\vol$.
Due to the reality property of $H$ the operators $\hat\Phi(\varphi)$ are hermitian for all $\varphi\in H$.
On the other hand, the elements of $\bigl(C^\infty_0(\mathcal{M},\mathbb{R})[[\lambda]],\widehat{\omega}_\star\bigr)$
should be interpreted as smearing functions, where the smearing is given by
$ \hat\Phi(\varphi)=\int_\mathcal{M}\hat\Phi\, \varphi\,\vol$. Thus, we can use both versions of the pre-symplectic
$\bbR[[\lambda]]$-modules in order to define a QFT, but one has to cope with different interpretations.
\end{remark}

\subsection{Symplectic isomorphism between the deformed and undeformed symplectic $\bbR[[\lambda]]$-module}
Consider the pre-symplectic vector space of the commutative field theory $\bigl(C^\infty_0(\mathcal{M},\mathbb{R}),\omega\bigr)$,
where
\begin{flalign}
\omega:C^\infty_0(\mathcal{M},\mathbb{R})\times C^\infty_0(\mathcal{M},\mathbb{R})\to\mathbb{R},~(\varphi,\psi)\mapsto
\omega(\varphi,\psi)=\int\limits_\mathcal{M}\varphi\,\Delta(\psi)\,\vol~.
\end{flalign}
Here $\Delta:=\Delta_+ -\Delta_-$ is the fundamental solution corresponding to the undeformed equation of motion operator
$P=\square_g-\xi\,\mathfrak{R}$. We now show that there exists a symplectic isomorphism between the formal power series
extension of $\bigl(C^\infty_0(\mathcal{M},\mathbb{R}),\omega\bigr)$, i.e.~the pre-symplectic $\bbR[[\lambda]]$-module
$\bigl(C^\infty_0(\mathcal{M},\mathbb{R})[[\lambda]],\omega\bigr)$, and 
$\bigl(C^\infty_0(\mathcal{M},\mathbb{R})[[\lambda]],\widehat{\omega}_\star\bigr)$.
\begin{propo}
\label{propo:presymplecto}
 The $\bbR[[\lambda]]$-linear map 
\begin{flalign}
\label{eqn:symplecto}
 S:C^\infty_0(\mathcal{M},\mathbb{R})[[\lambda]]\to C^\infty_0(\mathcal{M},\mathbb{R})[[\lambda]],~\varphi\mapsto S\varphi
=\sqrt{\cos\left(\frac{\lambda c}{4}(D+2)X_1\right)^{-1}}\varphi
\end{flalign}
has the following properties:
\begin{itemize}
 \item[$(i)$] $S$ is invertible
\item[$(ii)$] $\omega(S\varphi,S\psi)=\widehat{\omega}_\star(\varphi,\psi)$, for all 
$\varphi,\psi\in C^\infty_0(\mathcal{M},\mathbb{R})[[\lambda]]$
\item[$(iii)$] $S(\text{Ker}(\Delta))=S^{-1}(\text{Ker}(\Delta))=\text{Ker}(\Delta)$
\end{itemize}
\end{propo}
\begin{proof}
 $(i)$: The inverse of $S$ is given by $S^{-1}=\sqrt{\cos\left(\frac{\lambda c}{4}(D+2)X_1\right)}$.

$(ii)$: Let $\varphi,\psi\in C^\infty_0(\mathcal{M},\mathbb{R})[[\lambda]]$,
then
\begin{multline}
 \omega(S\varphi,S\psi)= \int\limits_\mathcal{M}S\varphi\,\Delta(S\psi)\,\vol~~ \stackrel{\text{PI}}{=} ~~
\int\limits_\mathcal{M}\varphi\,S\Bigl(\Delta(S\psi)\,\vol\Bigr)\\
 ~~\stackrel{\text{KP}}{=}~~\int\limits_\mathcal{M}\varphi\,S\Delta(S\psi)\,\vol ~~\stackrel{\text{KP}}{=}~~
\int\limits_\mathcal{M} \varphi\,\Delta(S^2\psi)\,\vol = \int\limits_\mathcal{M} \varphi\,\widehat{\Delta}_\star(\psi)\,\vol
=\widehat{\omega}_\star(\varphi,\psi)~.
\end{multline}
We have used integration by parts together with the fact that $S$ includes only even powers of $X_1$ (PI) 
and the Killing property of $X_1$, leading to $\mathcal{L}_{X_1}(\vol)=0$ and $[X_1,\Delta]=0$ (KP).

$(iii)$: We have $S^{\pm 1}(\text{Ker}(\Delta))\subseteq \text{Ker}(\Delta)$, since let $\varphi\in \text{Ker}(\Delta)$, then
$\Delta(S^{\pm1}\varphi)=S^{\pm1}\Delta(\varphi)=0$. We also have $S^{\pm1}(\text{Ker}(\Delta))\supseteq \text{Ker}(\Delta)$, since
let $\varphi\in\text{Ker}(\Delta)$, then there is $\psi=S^{\mp1}\varphi\in \text{Ker}(\Delta)$, such that
$\varphi=S^{\pm1}\psi$.
\end{proof}
Furthermore, it is easy to show that $\text{Ker}(\widehat{\Delta}_\star)=\text{Ker}(\Delta)$.
Let $\varphi\in\text{Ker}(\Delta)$, then
$\widehat{\Delta}_\star(\varphi)=\Delta(S^2\varphi)=S^2\Delta(\varphi)=0$. The other way around,
let $\varphi\in\text{Ker}(\widehat{\Delta}_\star)$, then $\Delta(\varphi)= S^{-2}\widehat{\Delta}_\star(\varphi)=0$.

This leads us to the following
\begin{cor}
\label{cor:symplecto}
Consider the symplectic $\bbR[[\lambda]]$-modules
$\bigl(\widehat{V}_\star,\widehat{\omega}_\star\bigr):= 
\bigl(C^\infty_0(\mathcal{M},\mathbb{R})[[\lambda]]/\text{Ker}(\widehat{\Delta}_\star),
\widehat{\omega}_\star\bigr)$ and $\bigl(V,\omega\bigr) := 
\bigl(C^\infty_0(\mathcal{M},\mathbb{R})[[\lambda]]/\text{Ker}(\Delta),\omega\bigr)$.
Then the map
\begin{flalign}
 \mathcal{S}:\widehat{V}_\star \to V,~[\varphi]\mapsto \mathcal{S}[\varphi]=[S\varphi]
\end{flalign}
is a symplectic isomorphism.
\end{cor}
\begin{proof}
 The map is well-defined, since $S(\text{Ker}(\widehat{\Delta}_\star))\subseteq \text{Ker}(\Delta)$. 
The inverse map is given by
\begin{flalign}
 \mathcal{S}^{-1}:V\to \widehat{V}_\star,~[\varphi]\mapsto \mathcal{S}^{-1}[\varphi] = [S^{-1}\varphi]~,
\end{flalign}
and is well-defined, too. Using Proposition \ref{propo:presymplecto} we find
$\omega(\mathcal{S}[\varphi],\mathcal{S}[\psi])= \widehat{\omega}_\star([\varphi],[\psi])$, for all $[\varphi],[\psi]\in 
\widehat{V}_\star$.
\end{proof}
Note that $\widehat{V}_\star=V$ as an $\bbR[[\lambda]]$-module.
In the following we will drop for notational reasons the brackets $[\,\cdot\,]$ indicating equivalence classes.

\subsection{Deformed $\ast$-algebra of field polynomials}
Making use of the deformed symplectic $\bbR[[\lambda]]$-module 
$\bigl(\widehat{V}_\star,\widehat{\omega}_\star\bigr)$ we can define 
the $\ast$-algebra of field observables for the NC QFT \cite{Ohl:2009qe}. Since our present focus is on
formal deformation quantization, the preferred choice of an algebra of observables is the $\ast$-algebra
of field polynomials (also called the Borchers-Uhlmann algebra) and not the Weyl algebra. 
This is due to the fact that formal power series prohibit the use of $C^\ast$-algebras.
For a review on algebras, states and representations in deformation quantization see \cite{waldmann}.
We make the following
\begin{defi}
\label{def:fieldpoly}
 Let $(W,\rho)$ be a symplectic $\bbR[[\lambda]]$-module.
 Let $\mathcal{A}_\text{free}$ be the unital $\ast$-algebra over $\mathbb{C}[[\lambda]]$ which is
 freely generated by the elements $1$, $\Phi( \varphi )$ and $\Phi(\varphi)^\ast$, $\varphi \in W$,
 and let $\mathcal{I}$ be the $\ast$-ideal generated by the elements
\begin{subequations}
\label{eqn:fieldalgebra}
\begin{flalign}
&~~\label{eqn:fieldoplin}\Phi(\gamma\,\varphi+\gamma^\prime\,\psi)-\gamma\,\Phi(\varphi)-\gamma^\prime\,\Phi(\psi)~,\\
&~~\Phi(\varphi)^\ast-\Phi(\varphi)~,\\
&~~[\Phi(\varphi),\Phi(\psi)]- i\,\rho(\varphi,\psi)\, 1~,
\end{flalign}
\end{subequations}
for all $\varphi,\psi\in W$ and $\gamma,\gamma^\prime\in \mathbb{R}[[\lambda]]$.
The $\ast$-algebra of field polynomials is defined as the quotient $\mathcal{A}_{(W,\rho)}:= \mathcal{A}_\text{free}/\mathcal{I}$.
\end{defi}
For our investigations we require the following canonical construction.
\begin{lem}
\label{lem:symplectoalgeba}
Let $(W_1,\rho_1)$ and $(W_2,\rho_2)$ be two symplectic $\bbR[[\lambda]]$-modules 
and let $S:W_1\to W_2$ be a symplectic isomorphism. Then the map
$\mathfrak{S}:\mathcal{A}_{(W_1,\rho_1)}\to \mathcal{A}_{(W_2,\rho_2)}$, which is defined on the generators by
\begin{subequations}
\begin{flalign}
 \mathfrak{S}(1)&=1~,\\
 \mathfrak{S}\bigl(\Phi(\varphi)\bigr)&=\Phi(S\varphi)~,~\forall \varphi\in W_1~,
\end{flalign}
\end{subequations}
and extended to $\mathcal{A}_{(W_1,\rho_1)}$ as a $\ast$-algebra homomorphism 
is a $\ast$-algebra isomorphism.
\end{lem}
\begin{proof}
It has to be shown that the relations (\ref{eqn:fieldalgebra}) are compatible with the map $\mathfrak{S}$.
By definition, the map $\mathfrak{S}$ is compatible with $\mathbb{R}[[\lambda]]$-linearity (\ref{eqn:fieldoplin}). We also find
\begin{flalign}
 \mathfrak{S}\bigl(\Phi(\varphi)\bigr)^\ast = \Phi(S\varphi)^\ast= \Phi(S\varphi) 
= \mathfrak{S}\bigl(\Phi(\varphi)\bigr)~,
\end{flalign}
and 
\begin{flalign}
\nonumber \mathfrak{S}\bigl(\bigl[\Phi(\varphi),\Phi(\psi)\bigr]\bigr)&=
\left[\mathfrak{S}\bigl(\Phi(\varphi)\bigr),\mathfrak{S}\bigl(\Phi(\psi)\bigr)\right] 
=\left[\Phi(S\varphi),\Phi(S\psi)\right]\\
&= i\,\rho_2(S\varphi,S\psi)\,1=
\mathfrak{S}\bigl( i\,\rho_1(\varphi,\psi)\,1 \bigr)~, 
\end{flalign}
for all $\varphi,\psi\in W_1$. Thus, the map $\mathfrak{S}$ can be extended  to a well-defined 
$\ast$-algebra homomorphism between $\mathcal{A}_{(W_1,\rho_1)}$ and $\mathcal{A}_{(W_2,\rho_2)}$.
$\mathfrak{S}$ is invertible, as shown by using the inverse $S^{-1}$ of $S$, and thus is a
$\ast$-algebra isomorphism between $\mathcal{A}_{(W_1,\rho_1)}$ and $\mathcal{A}_{(W_2,\rho_2)}$.
\end{proof}

Consider the $\ast$-algebras of field polynomials of the undeformed and the deformed 
QFT given by $\mathcal{A}_{(V,\omega)}$ and $\mathcal{A}_{(\widehat{V}_\star,\widehat{\omega}_\star)}$,
respectively. We obtain by applying Lemma \ref{lem:symplectoalgeba} and Corollary \ref{cor:symplecto}
\begin{cor}
\label{cor:fieldpolyiso}
There exists a $\ast$-algebra isomorphism $\mathfrak{S}:\mathcal{A}_{(\widehat{V}_\star,\widehat{\omega}_\star)}
\to\mathcal{A}_{(V,\omega)}$.
\end{cor}
 Note that this corollary states that we can {\it mathematically} describe the NC QFT in terms of a 
formal power series extension of the corresponding commutative QFT.
 However, the {\it physical interpretation} has to be adapted properly: If we want to probe the 
NC QFT with a set of smearing functions $\lbrace \varphi_i:i\in\mathcal{I}\rbrace$ in order to extract
physical observables (e.g.~Wightman functions) we have to probe the commutative QFT with a different set of smearing 
functions $\lbrace S\varphi_i: i\in\mathcal{I} \rbrace$, with $S$ given in (\ref{eqn:symplecto}).
\begin{remark}
 In algebraic QFT the main object of interest is the net of algebras of local observables \cite{Haag:1992hx,Araki:1999ar}.
 See also \cite{baer} for the construction of a net of local algebras for commutative, free and bosonic QFTs.
 The symplectic isomorphism $\mathcal{S}$ of Corollary \ref{cor:symplecto} is a finite order differential operator
 at every order in $\lambda$. Thus, 
 in the framework of formal power series, the relation between the undeformed and deformed QFT is local, meaning that there is
 an isomorphism between the corresponding nets.
 This feature does not carry over to convergent deformations, since there the symplectic isomorphism is
 a nonlocal map, see Section \ref{sec:convergent}.
\end{remark}

\subsection{Symmetries of the deformed and undeformed QFT}
We study relations among the symmetries of the $\ast$-algebra of field polynomials of the undeformed and deformed
QFT. The symmetries we consider are of the following type:
\begin{defi}
\label{def:symplecticauto}
 Let $(W,\rho)$ be a symplectic $\bbR[[\lambda]]$-module. The set of
 symplectic automorphisms $\mathcal{G}_{(W,\rho)}\subseteq \text{End}_{\mathbb{R}[[\lambda]]}(W)$, 
i.e.~the set of all invertible $\alpha\in\text{End}_{\mathbb{R}[[\lambda]]}(W)$ satisfying
$\rho\bigl(\alpha\varphi,\alpha\psi\bigr)=\rho(\varphi,\psi)$, for all $\varphi,\psi\in W$, 
forms a group under the usual composition of homomorphisms $\circ$.
This group is called the group of symplectic automorphisms.
\end{defi}
For completeness we state (without proof) two simple lemmas which are important for the following investigations.
\begin{lem}
 Let $(W_1,\rho_1)$ and $(W_2,\rho_2)$ be two symplectic $\bbR[[\lambda]]$-modules
and let $S:W_1\to W_2$ be a symplectic isomorphism. Then the map
\begin{flalign}
S_\mathcal{G}:\text{End}_{\mathbb{R}[[\lambda]]}(W_2)\to\text{End}_{\mathbb{R}[[\lambda]]}(W_1),
~\alpha\mapsto S^{-1}\circ \alpha\circ S
\end{flalign}
provides a group isomorphism between $\mathcal{G}_{(W_2,\rho_2)}$ and $\mathcal{G}_{(W_1,\rho_1)}$.
\end{lem}
\begin{lem}
The group $\mathcal{G}_{(W,\rho)}$ can be represented on the $\ast$-algebra of field polynomials $\mathcal{A}_{(W,\rho)}$ 
by defining the action on the generators
\begin{subequations}
\begin{flalign}
 \pi_\alpha(1)&=1~,\\
\pi_\alpha\bigl(\Phi(\varphi)\bigr)&=\Phi\bigl(\alpha\varphi\bigr)~,~\forall\varphi\in W~,
\end{flalign}
\end{subequations}
for all $\alpha\in \mathcal{G}_{(W,\rho)}$, and extending to $\mathcal{A}_{(W,\rho)}$ as $\ast$-algebra homomorphisms.
\end{lem}
We denote by $\mathcal{G}:=\mathcal{G}_{(V,\omega)}$ and $\widehat{\mathcal{G}}_\star:=
\mathcal{G}_{(\widehat{V}_\star,\widehat{\omega}_\star)}$ the group of symplectic automorphisms 
of the symplectic $\bbR[[\lambda]]$-modules $(V,\omega)$ and $\bigl(\widehat{V}_\star,\widehat{\omega}_\star\bigr)$, 
respectively.
From the lemmas above and Corollary \ref{cor:symplecto} we obtain
\begin{cor}
\label{cor:symplecticgroup}
There exists a group isomorphism $S_\mathcal{G}:\mathcal{G}\to\widehat{\mathcal{G}}_\star
,~\alpha\mapsto \mathcal{S}^{-1}\circ 
\alpha\circ \mathcal{S}$.
\end{cor}
 Since, as $\bbR[[\lambda]]$-modules, $\widehat{V}_\star$ and $V$ are identical, 
we can make sense of acting with $\alpha\in\mathcal{G}$
on $\widehat{V}_\star$. Thus, we can write $S_\mathcal{G}(\alpha) = \mathcal{S}^{-1}\circ\alpha\circ \mathcal{S}\circ
 \alpha^{-1}\circ \alpha =: \mathcal{S}^{-1}\circ \mathcal{S}_\alpha\circ \alpha$, for all $\alpha\in\mathcal{G}$.
In case the symplectic isomorphism $\mathcal{S}$ is invariant under the adjoint action of
$\alpha$, i.e.~$\mathcal{S}_\alpha=\alpha\circ\mathcal{S}\circ \alpha^{-1}=\mathcal{S}$,
 we simply obtain the standard transformation law 
$\mathcal{S}_\mathcal{G}(\alpha)=\alpha$ in the deformed field theory.  In the case of 
$\mathcal{S}_\alpha\neq\mathcal{S}$, the undeformed transformation $\alpha$ is implemented in a non-canonical
way on the deformed symplectic $\bbR[[\lambda]]$-module and also on the deformed algebra of field polynomials.
Thus, the corollary above shows that the deformed field theory enjoys the same amount of symmetries as the undeformed one, 
with the difference that some transformations (what is usually called the broken symmetries) 
are represented in a non-canonical way.
\begin{remark}
 By using explicit toy-models we show in Section \ref{sec:convergent} that
 Corollary \ref{cor:symplecticgroup} does not carry over to convergent deformations.
 This means that convergent deformations can break symmetries, which were not broken before in 
 formal deformation quantization.
\end{remark}

\subsection{States on the deformed $\ast$-algebra of field polynomials}
Let us first fix notation. A state on a unital $\ast$-algebra $\mathcal{A}$ over $\mathbb{C}[[\lambda]]$ 
is a $\mathbb{C}[[\lambda]]$-linear map $\Omega:\mathcal{A}\to\mathbb{C}[[\lambda]]$, such that
\begin{subequations}
\begin{flalign}
 \Omega(1)&=1~,\\
 \Omega(a^\ast\,a)&\geq 0~,~\forall a\in\mathcal{A}~.
\end{flalign}
\end{subequations}
The ordering on $\mathbb{R}[[\lambda]]$ is defined by
\begin{flalign}
 \gamma= \sum\limits_{n=n_0}^{\infty}\lambda^n \gamma_n>0\quad :\Longleftrightarrow\quad \gamma_{n_0}>0~.
\end{flalign}
Assume that there is a group $G$ acting on $\mathcal{A}$ as $\ast$-algebra automorphisms.
We say that a state $\Omega$ is symmetric under $G$, if
\begin{flalign}
\Omega\bigl(\pi_\alpha (a)\bigr) =\Omega(a)~,
\end{flalign}
for all $a\in\mathcal{A}$ and $\alpha\in G$.

We remind the reader of the following standard construction.
\begin{lem}
\label{lem:pullback}
 Let $\mathcal{A}_1$ and $\mathcal{A}_2$ be two unital $\ast$-algebras over $\mathbb{C}[[\lambda]]$ and let 
$\kappa:\mathcal{A}_1 \to \mathcal{A}_2$ be a $\ast$-algebra homomorphism. Then each state $\Omega_2$ on $\mathcal{A}_2$
induces a state $\Omega_1$ on $\mathcal{A}_1$ by defining
\begin{flalign}
 \Omega_1(a):= \Omega_2\bigl(\kappa (a)\bigr)~,
\end{flalign}
for all $a\in\mathcal{A}_1$.
\end{lem}
\begin{proof}
 $\Omega_1:\mathcal{A}_1\to\mathbb{C}[[\lambda]]$ is a $\mathbb{C}[[\lambda]]$-linear map.
We have $\Omega_1(1)=\Omega_2(\kappa(1))=\Omega_2(1)=1$ and
\begin{flalign}
\Omega_1(a^\ast\,a)=\Omega_2(\kappa(a^\ast\,a))=\Omega_2(\kappa(a)^\ast\,\kappa(a))\geq 0~,
\end{flalign}
for all $a\in\mathcal{A}_1$.
\end{proof}
The state $\Omega_1$ defined above is called the pull-back of $\Omega_2$ under the $\ast$-algebra homomorphism $\kappa$.
Note that in case of a $\ast$-algebra isomorphism $\kappa:\mathcal{A}_1\to\mathcal{A}_2$ 
there is a bijection between the states on $\mathcal{A}_1$ and the states on
$\mathcal{A}_2$.

For the case of the deformed and undeformed QFT this leads us to the following
\begin{cor}
There is a bijection between the states on $\mathcal{A}_{(V,\omega)}$ and the states on
$\mathcal{A}_{(\widehat{V}_\star,\widehat{\omega}_\star)}$. 
The $\mathcal{G}$-symmetric states on $\mathcal{A}_{(V,\omega)}$ are pulled-back 
to $\widehat{\mathcal{G}}_\star$-symmetric states on $\mathcal{A}_{(\widehat{V}_\star,\widehat{\omega}_\star)}$, and vice versa.
\end{cor} 
\begin{proof}
 It remains to show that symmetric states are pulled-back to symmetric states.
 Let $\Omega$ be a $\mathcal{G}$-symmetric state on  $\mathcal{A}_{(V,\omega)}$ and
 let $\alpha\in\widehat{\mathcal{G}}_\star$ be arbitrary.
 By Corollary \ref{cor:symplecticgroup} we find a $\beta\in\mathcal{G}$, such that 
$\alpha=\mathcal{S}^{-1}\circ\beta\circ\mathcal{S}$. We obtain for all $a\in\mathcal{A}_{(\widehat{V}_\star,\widehat{\omega}_\star)}$
\begin{flalign}
 \Omega_\star\bigl(\pi_\alpha (a)\bigr)= \Omega\bigl(\mathfrak{S}(\pi_\alpha (a))\bigr)=\Omega\bigl((\mathfrak{S}\circ
\mathfrak{S}^{-1}\circ\pi_\beta\circ\mathfrak{S})(a)  \bigr)=\Omega\bigl(\pi_\beta(\mathfrak{S}(a))\bigr)
= \Omega_\star(a)~.
\end{flalign}
The vice versa is proven analogously.
\end{proof}


\section{\label{sec:convergent}Examples of convergent homothetic Killing deformations}
\subsection{\label{ex:frwqft}A spatially flat FRW toy-model}
In this section we apply the formalism developed in Section \ref{sec:formal} to a toy-model.
We use a special choice of the FRW spacetime of Example \ref{ex:frw}. Let $\mathcal{M}=(0,\infty)\times\mathbb{R}^3$, i.e.~$D=4$,
and let $t$ and $x^i,~i\in\lbrace1,2,3\rbrace,$ be global coordinates.
The metric field we consider is given by $g=-dt\otimes dt + t^2\,\delta_{ij}dx^i\otimes dx^j$.
Note that in our conventions the spatial coordinates $x^i$ are dimensionless. The reason for choosing
the scale factor $a(t)\propto t$ is that in this case a proper homothetic Killing vector field
is given by $X_2=t\partial_t$ and {\it all} Killing vector fields commute with $X_2$.
The most general real Killing vector field
is  $k_{(\xi,\eta)} :=\xi^i\partial_i + \eta^{k}\epsilon_{kij}x^i\partial_j$, where $\xi,\eta\in\mathbb{R}^3$.

\paragraph{The undeformed theory:}
We start by collecting useful formulae of the undeformed free, real, massless and curvature coupled scalar QFT on our particular
FRW spacetime. They will be used later to study the NC deformation. We frequently
use the Fourier transformation on the spatial hypersurfaces $\mathbb{R}^3$ defined by $t=\text{const}$.
We indicate this transformation by a tilde and use the conventions
\begin{flalign}
 \widetilde{\varphi}(t,k)= \int\limits_{\mathbb{R}^3}d^3x\,e^{ikx}\,\varphi(t,x)\quad ,\qquad \varphi(t,x)=
\int\limits_{\mathbb{R}^3}\frac{d^3k}{(2\pi)^3}\,e^{-ikx}\,\widetilde{\varphi}(t,k)~.
\end{flalign}
The wave operator $P=\square_g-\xi\,\mathfrak{R}$ in Fourier space is given by
\begin{flalign}
 \widetilde{P}\bigl(\widetilde{\varphi}\bigr)(t,k) = -\left(\partial_t^2+\frac{3}{t}\partial_t + \frac{k^2+6\xi}{t^2}\right)\widetilde{\varphi}(t,k)~.
\end{flalign}
The corresponding retarded and advanced Green's operators read
\begin{flalign}
\label{eqn:greenfrw}
 \widetilde{\Delta}_{\pm}\bigl(\widetilde{\varphi}\bigr)(t,k)=-\int\limits_{t_\pm}^t d\tau\tau^3\,\widetilde{\Delta}(t,\tau,k)\,
\widetilde{\varphi}(\tau,k)~,
\end{flalign}
where $t_+=0$, $t_-=\infty$ and
\begin{flalign}
\widetilde{\Delta}(t,\tau,k)= \frac{t^{\sqrt{1-k^2-6\xi}}\tau^{-\sqrt{1-k^2-6\xi}}-
t^{-\sqrt{1-k^2-6\xi}}\tau^{\sqrt{1-k^2-6\xi}}}{2t\tau\sqrt{1-k^2-6\xi}}~.
\end{flalign}
We obtain for the fundamental solution $\Delta=\Delta_+-\Delta_-$
\begin{flalign}
  \widetilde{\Delta}\bigl(\widetilde{\varphi}\bigr)(t,k)=-\int\limits_{0}^{\infty} d\tau\tau^3\,\widetilde{\Delta}(t,\tau,k)\,
\widetilde{\varphi}(\tau,k)~,
\end{flalign}
resulting in the pre-symplectic structure
\begin{flalign}
 \omega(\varphi,\psi) = -\int\limits_{0}^\infty dt t^3 \int\limits_{0}^\infty d\tau \tau^3\int\limits_{\mathbb{R}^3} 
\frac{d^3k}{(2\pi)^3}\,\widetilde{\varphi}(t,-k)\,\widetilde{\Delta}(t,\tau,k)\,\widetilde{\psi}(\tau,k)~.
\end{flalign}

We define the geometric action of $(R,a)\in SO(3)\ltimes\mathbb{R}^3$ on $C^\infty(\mathcal{M})$ by
\begin{flalign}
 \bigl(\alpha_{(R,a)} \varphi\bigr)(t,x) :=\varphi\bigl(t,R^{-1}(x-a)\bigr)~.
\end{flalign}
In Fourier space, these transformations are given by
\begin{flalign}
\label{eqn:geometricactionfrw}
 \bigl(\widetilde{\alpha}_{(R,a)}\widetilde{\varphi}\bigr)(t,k)=e^{ika}\,\widetilde{\varphi}(t,R^{-1}k)~.
\end{flalign}
We easily obtain that $SO(3)\ltimes\mathbb{R}^3\subseteq \mathcal{G}_{(V,\omega)}$ are symplectic automorphisms
of the symplectic vector space $(V,\omega)=
\bigl(C^\infty_0(\mathcal{M},\mathbb{R})/P[C^\infty_0(\mathcal{M},\mathbb{R})],\omega\bigr)$.

In this section we are working in a convergent framework and all symplectic vector spaces are vector spaces 
over $\mathbb{R}$. Thus, we can apply the powerful theory of Weyl algebras in order to quantize the
symplectic vector space $(V,\omega)$, i.e.~in order to define the QFT \cite{baer}. 
We briefly fix notation. A Weyl system $(A,W)$ of a symplectic vector space $(V,\omega)$ 
 consists of a unital $C^\ast$-algebra $A$ and a map $W:V\to A$ such 
that for all $\varphi,\psi\in V$ we have
\begin{subequations}
\begin{flalign}
 W(0)&=1~,\\
 W(\varphi)^\ast &= W(-\varphi)~,\\
 W(\varphi)\,W(\psi)&=e^{-i\omega(\varphi,\psi)/2}\, W(\varphi+\psi)~.
\end{flalign}
\end{subequations}
A Weyl system $(A,W)$ of a symplectic vector space $(V,\omega)$ is 
called a CCR-representation of $(V,\omega)$ if $A$ is generated 
as a $C^\ast$-algebra by the elements $W(\varphi),~\varphi\in V$. The algebra $A$
is called the Weyl algebra (or CCR algebra) and it is unique up to $\ast$-isomorphisms.
The group of symplectic automorphisms of $(V,\omega)$ can be represented on the Weyl algebra
by employing Corollary 4.2.11 of \cite{baer}, which states that given a symplectic linear map $S:V_1\to V_2$ 
between two symplectic vector spaces, there exists a unique injective $\ast$-morphism
$\mathfrak{S}:A_1\to A_2$ between the corresponding Weyl algebras such that $W_2(S\varphi)=\mathfrak{S}\bigl(W_1(\varphi)\bigr)$,
for all $\varphi\in V_1$.

The theory is fixed after we have specified an algebraic state $\Omega: A\to\mathbb{C}$.
This choice is in general highly nonunique. In our work we do not require an explicit choice
of state and we just assume that $\Omega$ is an $SO(3)\ltimes\mathbb{R}^3$-invariant state.
From the physical perspective, it is also natural to impose the Hadamard condition, regularity and quasifreeness.

\paragraph{The deformed theory with ${\mathbf X_1=\partial_1}$:}
We study a convergent deformation of our FRW model. We choose $X_1=\partial_1$, 
i.e.~a translation along the $x^1$-direction. 
The flow generated by $X_1=\partial_1$ is noncompact.
The condition $[X_1,X_2]=0$, which is required for our deformations, is satisfied. 
Our strategy is to make a convergent definition of the maps 
$S=\sqrt{\cos(3\lambda \partial_1)^{-1}}$ and $S^{-1}=\sqrt{\cos(3\lambda \partial_1)}$,
which enter the construction of the deformed QFT in Section \ref{sec:formal}. Using these
maps we construct the deformed QFT and investigate its properties. 

We first define a convenient space of functions $C^\infty_{0,\mathscr{S}}(\mathcal{M})\subset C^\infty(\mathcal{M})$,
where $\mathcal{M}=(0,\infty)\times \mathbb{R}^3$. 
A function $\varphi\in C^\infty(\mathcal{M})$ is in $C^\infty_{0,\mathscr{S}}(\mathcal{M})$, iff
$(1.)$ for all fixed $x\in\mathbb{R}^3$ $\varphi(t,x)\in C^\infty_0((0,\infty))$
and $(2.)$ for all fixed $t\in (0,\infty)$ $\varphi(t,x)\in \mathscr{S}(\mathbb{R}^3)$ is a Schwartz function.
We have $C^\infty_0(\mathcal{M})\subset C^\infty_{0,\mathscr{S}}(\mathcal{M})\subset C^\infty(\mathcal{M})$.
The spatial Fourier transformation
is an automorphism of $ C^\infty_{0,\mathscr{S}}(\mathcal{M})$, i.e.~let $\varphi(t,x)\in C^\infty_{0,\mathscr{S}}(\mathcal{M})$
 then $\widetilde{\varphi}(t,k)\in C^\infty_{0,\mathscr{S}}(\mathcal{M})$ and vice versa.

Using the spatial Fourier transformation we define the map $S:C^\infty_{0,\mathscr{S}}(\mathcal{M})\to 
 C^\infty_{0,\mathscr{S}}(\mathcal{M})$ by
\begin{flalign}
\label{eqn:convsymp}
 \bigl(\widetilde{S}\widetilde{\varphi}\bigr)(t,k) := \sqrt{\cosh(3\lambda k_1)^{-1}}\,\widetilde{\varphi}(t,k)~.
\end{flalign}
The hyperbolic cosine in (\ref{eqn:convsymp}) is due $\partial_i\to -ik_i$ in Fourier space.  
This map is injective. The inverse map $S^{-1}:S\bigl(C^\infty_{0,\mathscr{S}}(\mathcal{M})\bigr)\to 
C^\infty_{0,\mathscr{S}}(\mathcal{M})$ is given by
\begin{flalign}
 \bigl(\widetilde{S}^{-1}\widetilde{\varphi}\bigr)(t,k) := \sqrt{\cosh(3\lambda k_1)}\,\widetilde{\varphi}(t,k)~.
\end{flalign}
Note that $S\bigl(C^\infty_{0,\mathscr{S}}(\mathcal{M})\bigr)\subset C^\infty_{0,\mathscr{S}}(\mathcal{M})$, since
$S\bigl(C^\infty_{0,\mathscr{S}}(\mathcal{M})\bigr)$ includes only functions with Fourier spectra
decreasing faster than $e^{-3\lambda \vert k_1\vert /2}$ for large $\vert k_1\vert$.

We can now construct the deformed QFT. We consider only the hatted quantities of
Section \ref{sec:formal}, 
i.e.~$\widehat{P}_\star$, $\widehat{\Delta}_{\star\pm}$, $\widehat{\omega}_\star$, etc.~and 
drop all hats to simplify notation. We define the deformed Green's operators by
\begin{flalign}
 \widetilde{\Delta}_{\star\pm}\bigl(\widetilde{\varphi}\bigr)(t,k):=\widetilde{\Delta}_{\pm}\bigl(\widetilde{S}^2\widetilde{\varphi}\bigr)(t,k)
=-\int\limits_{t_\pm}^t d\tau\tau^3\,\frac{\widetilde{\Delta}(t,\tau,k)}{\cosh(3\lambda k_1)}\,\widetilde{\varphi}(\tau,k)~.
\end{flalign}
This results in the deformed fundamental solution
\begin{flalign}
\label{eqn:frwdefsympl}
 \widetilde{\Delta}_{\star}\bigl(\widetilde{\varphi}\bigr)(t,k)=-\int\limits_{0}^\infty d\tau\tau^3\,
\frac{\widetilde{\Delta}(t,\tau,k)}{\cosh(3\lambda k_1)}\,\widetilde{\varphi}(\tau,k)~,
\end{flalign}
and the deformed pre-symplectic structure on $C^\infty_0(\mathcal{M},\mathbb{R})$
\begin{flalign}
 \omega_\star(\varphi,\psi) = -\int\limits_{0}^\infty dt t^3 \int\limits_{0}^\infty d\tau \tau^3\int\limits_{\mathbb{R}^3} 
\frac{d^3k}{(2\pi)^3}\,\widetilde{\varphi}(t,-k)\,\frac{\widetilde{\Delta}(t,\tau,k)}{\cosh(3\lambda k_1)}\,
\widetilde{\psi}(\tau,k)~.
\end{flalign}
The deformed fundamental solution satisfies $\text{Ker}(\Delta_\star) = \text{Ker}(\Delta)$.
To prove this, let $\varphi\in C^\infty_0(\mathcal{M},\mathbb{R})$. Using (\ref{eqn:frwdefsympl}) we find
\begin{flalign}
\widetilde{\Delta}_\star\bigl(\widetilde{\varphi}\bigr)(t,k)=\cosh(3\lambda k_1)^{-1}\,\widetilde{\Delta}\bigl(\widetilde{\varphi}\bigr)(t,k)~,
\end{flalign}
and the proof follows from the positivity of $\cosh(3\lambda k_1)^{-1}$.
From general considerations \cite{baer} we know that $\text{Ker}(\Delta)=P[C^\infty_0(\mathcal{M},\mathbb{R})]$.
Thus, we can define the deformed symplectic vector space as $(V_\star,\omega_\star):= 
\bigl(C^\infty_0(\mathcal{M},\mathbb{R})/P[C^\infty_0(\mathcal{M},\mathbb{R})], \omega_\star\bigr)$.
Different to Section \ref{sec:formal} this is now a vector space over $\mathbb{R}$ and not a module over 
the ring $\mathbb{R}[[\lambda]]$.

The construction of the deformed QFT in terms of a CCR-representation of $(V_\star,\omega_\star)$
can be performed analogously to the undeformed QFT, since $(V_\star,\omega_\star)$ is symplectic 
vector space over $\mathbb{R}$.
This results in a unique (up to $\ast$-isomorphisms) deformed Weyl algebra $A_\star$
describing the deformed QFT.

\paragraph{Physical features of the deformed theory:}
Let us investigate some physical features of the deformed field theory. In the following we assume that $\lambda>0$.
We study in more detail the map $S^2$ acting on $C^\infty_0(\mathcal{M})$. Note that in 
position space this map is given by the following convolution
\begin{flalign}
\label{eqn:convolution}
 \bigl(S^2\varphi\bigr)(t,x) = \int\limits_\mathbb{R}dy^1\,\frac{1}{6\lambda\,\cosh\bigl(\pi (x^1-y^1)/6\lambda\bigr)}\,
\varphi(t,y^1,x^2,x^3)~,
\end{flalign}
for all $\varphi\in C^\infty_0(\mathcal{M})$. It is easy to check that the image $S^2[C^\infty_0(\mathcal{M})]\not\subseteq 
C^\infty_0(\mathcal{M})$. For this assume that $\varphi\in C^\infty_0(\mathcal{M})$ is a positive semidefinite function localized
in some compact region $K\subset\mathcal{M}$, e.g.~a bump function. Since the convolution kernel is of noncompact support
and strictly positive, the resulting function $S^2\varphi$ is of noncompact support in the $x^1$-direction.
However, $S^2\varphi$ will be of rapid decrease, since the convolution kernel is a Schwartz function.

Physically, this means that causality is lost. We immediately obtain that the  relation
$\Delta_{\star\pm}(\varphi)\subseteq J^\pm\bigl(\text{supp}(\varphi)\bigr)$, for all $\varphi\in C^\infty_0(\mathcal{M})$, 
is violated. 
Thus, external sources couple in a nonlocal way to our deformed field theory, what is a feature not
present in the commutative counterpart. Since $\Delta_{\star\pm}$ depends on the value of $\lambda$ through (\ref{eqn:convolution}),
we can determine $\lambda$ (in principle) by measuring the response of the field to external excitations.

Consider now the deformed symplectic structure in position space
\begin{flalign}
\label{eqn:nonlocalsymp}
 \omega_\star(\varphi,\psi)=\int\limits_\mathcal{M} \varphi\,\Delta(S^2\psi)\,\vol~.
\end{flalign}
Due to the appearance of the nonlocal map $S^2$ (\ref{eqn:convolution}) there
are functions $\varphi,\psi$ with causally disconnected support satisfying $\omega_\star(\varphi,\psi)\neq 0$. 
For our choice of deformation ($X_1=\partial_1$ and $X_2=t\partial_t$)  this nonlocality 
only affects the $x^1$-direction, but its range is infinite. We still obtain that $\omega_\star(\varphi,\psi)=0$ if
arbitrary translations of  $\text{supp}(\psi)$ along the $x^1$-direction and $\text{supp}(\varphi)$ are
causally disconnected.
In the QFT described by $A_\star$ this nonlocal behavior leads to a noncommutativity between observables located in
 spacelike separated regions in spacetime.

\paragraph{$\ast$-isomorphism between the deformed and a nonstandard undeformed QFT:}
We have shown in the previous paragraph that 
$S^2[C^\infty_0(\mathcal{M},\mathbb{R})]\not\subseteq C^\infty_0(\mathcal{M},\mathbb{R})$.
The map $S$ defined in (\ref{eqn:convsymp}) can also be rewritten in terms of a convolution in position
space, similar to $S^2$, but the convolution kernel is more complicated. Its explicit form is not of importance to us.

From $S^2[C^\infty_0(\mathcal{M},\mathbb{R})]\not\subseteq C^\infty_0(\mathcal{M},\mathbb{R})$ we can infer
that $S[C^\infty_0(\mathcal{M},\mathbb{R})]\not\subseteq C^\infty_0(\mathcal{M},\mathbb{R})$. To prove this, assume
that $S[C^\infty_0(\mathcal{M},\mathbb{R})]\subseteq C^\infty_0(\mathcal{M},\mathbb{R})$ then we would find
$S^2[C^\infty_0(\mathcal{M},\mathbb{R})]\subseteq C^\infty_0(\mathcal{M},\mathbb{R})$ by applying $S$ twice, 
which contradicts the observation above.

The bijective map $S:C^\infty_0(\mathcal{M},\mathbb{R})\to C^\infty_\text{img}(\mathcal{M},\mathbb{R})
\subset C^\infty_{0,\mathscr{S}}(\mathcal{M},\mathbb{R})$ induces a symplectic isomorphism between
$(V_\star,\omega_\star)$ and $(V_{\text{img}},\omega):=\bigl(C^\infty_\text{img}(\mathcal{M},\mathbb{R})
/S\bigl[P[C_0^\infty(\mathcal{M},\mathbb{R})]\bigr], \omega\bigr)$, i.e.~the deformed field theory
can be related to a nonstandard undeformed one. Using again Corollary 4.2.11 of \cite{baer},
this map induces a unique $\ast$-isomorphism between the Weyl algebras $A_\star$ and 
$A_\text{img}$.

Let us consider the symmetries of the deformed QFT. As we have seen above, the undeformed symplectic
vector space $(V,\omega)$ contains $SO(3)\ltimes \mathbb{R}^3$ in the group of symplectic automorphisms.
The representation is given by the geometric action (\ref{eqn:geometricactionfrw}).
However, the space $C^\infty_{\text{img}}(\mathcal{M},\mathbb{R})\subset C^\infty_{0,\mathscr{S}}(\mathcal{M},\mathbb{R})$,
 which serves as a pre-symplectic vector space for $(V_\text{img},\omega)$, is not invariant under the action of 
$SO(3)\ltimes \mathbb{R}^3$. The preserved subgroup is $SO(2)\ltimes \mathbb{R}^3$, where 
the $SO(2)$-rotation preserves the $x^1$-axis. Thus, the deformed QFT $A_\star$ is $\ast$-isomorphic
 to a nonsymmetric undeformed QFT $A_\text{img}$. This, in particular, shows that
Corollary \ref{cor:symplecticgroup} is restricted to formal deformation quantization.

\paragraph{Physical interpretation:}
As shown above, we do not have a $\ast$-isomorphism between the deformed QFT $A_\star$ and the {\it standard} undeformed 
QFT $A$ for convergent deformations. This fact has a very natural physical interpretation, which we will 
explain now.

Consider the extended symplectic vector space
$(V_\text{ext},\omega):= \bigl(C^\infty_{0,\mathscr{S}}(\mathcal{M},\mathbb{R})/\text{Ker}(\Delta),\omega\bigr)$.
This vector space carries a representation of $SO(3)\ltimes\mathbb{R}^3$ via the geometric
action (\ref{eqn:geometricactionfrw}). 
A CCR-representation of $(V_\text{ext},\omega)$ yields the extended Weyl algebra $A_\text{ext}$.
$A_\text{ext}$ carries a representation of the group $SO(3)\ltimes\mathbb{R}^3$. 
The algebra $A_\text{ext}$ is an extension of the usual Weyl algebra $A$, where also certain delocalized observables
are allowed. This already shows that $A_\text{ext}$ is more suitable to study the connection between deformed and undeformed
QFTs.

Analogously to the construction above, we define the deformed extended symplectic vector space
 $(V_{\text{ext}\star},\omega_\star):=(V_\text{ext},\omega_\star)$ and 
a CCR-representation of $(V_{\text{ext}\star},\omega_\star)$ yields the deformed extended Weyl algebra $A_{\text{ext}\star}$.

The bijective map $S:C^\infty_{0,\mathscr{S}}(\mathcal{M},\mathbb{R})\to C^\infty_{\text{img}^\prime}(\mathcal{M},\mathbb{R})
\subset C^\infty_{0,\mathscr{S}}(\mathcal{M},\mathbb{R})$ induces a symplectic isomorphism between
$(V_{\text{ext}\star},\omega_\star)$ and $(V_{\text{img}^\prime},\omega)$. Furthermore,
$S$ induces a symplectic embedding $(V_{\text{ext}\star},\omega_\star)\rightarrow (V_\text{ext},\omega)$.
Using Corollary 4.2.11 of \cite{baer}
we can induce a unique $\ast$-isomorphism between the Weyl algebras $A_{\text{ext}\star}$ and 
$A_{\text{img}^\prime}$, as well as a unique injective, but not surjective, $\ast$-morphism from $A_{\text{ext}\star}$ 
to $A_\text{ext}$. We thus have $A_{\text{ext}\star}\simeq A_{\text{img}^\prime}\subset A_{\text{ext}}$.

Physically, this means that due to the NC deformation the QFT looses observables. 
Since $C^\infty_{\text{img}^\prime}(\mathcal{M},\mathbb{R})$ depends on the value of 
$\lambda$, the observables we loose also depend on the value of $\lambda$.
Note that all functions in $C^\infty_{\text{img}^\prime}(\mathcal{M},\mathbb{R})$ have Fourier spectra which decrease
 faster than $e^{-3\lambda \vert k_1\vert/2}$ for large $\vert k_1\vert$. Increasing $\lambda$ leads 
to a sharper localization in momentum space and therefore 
a weaker localization in position space. Thus, the deformed QFT looses those observables
that are strongly localized in position space, which is physically very natural.

\paragraph{Induction of states:}
The injective $\ast$-morphism $\mathfrak{S}:A_\star\to A_{\text{ext}}$, 
or its extension $\mathfrak{S}_\text{ext}:A_{\text{ext}\star}\to A_{\text{ext}}$,
is useful for inducing states via the pull-back, see Lemma \ref{lem:pullback}.
Assume that we have an $SO(3)\ltimes\mathbb{R}^3$-invariant state $\Omega_{\text{ext}}:A_\text{ext}\to \mathbb{C}$.
Defining $\Omega_\star:= \Omega_\text{ext}\circ\mathfrak{S}:A_\star\to\mathbb{C}$ induces a state
on $A_\star$, which is invariant under the unbroken symmetry group $SO(2)\ltimes\mathbb{R}^3$.
The same holds true for $\Omega_{\text{ext}\star}:= \Omega_\text{ext}\circ\mathfrak{S}_\text{ext}:A_{\text{ext}\star}\to\mathbb{C}$.

\paragraph{NC effects on the cosmological power spectrum:}
In order to show that our models lead to nontrivial physical effects we briefly discuss the
cosmological power spectrum. We shall work with the extended Weyl algebras $A_{\text{ext}\star}$ and $A_{\text{ext}}$.
Let $\Omega:A_{\text{ext}}\to \mathbb{C}$ be a regular, quasifree and translation invariant (i.e.~$\mathbb{R}^3$-invariant)
state. We denote by $\bigl(\mathcal{H},\pi,\vert 0\rangle\bigr)$ the GNS-representation of $(A_{\text{ext}},\Omega)$, 
where $\mathcal{H}$ is a Hilbert space with scalar product $\langle\cdot\vert\cdot\rangle$,
$\pi:A_{\text{ext}}\to \mathcal{B}(\mathcal{H})$ is a representation in terms of bounded operators and 
$\vert 0\rangle\in\mathcal{H}$ is a cyclic vector of norm $1$.
Employing the $\ast$-morphism $\mathfrak{S}_\text{ext}:A_{\text{ext}\star}\to A_{\text{ext}}$ we obtain
a representation $\pi_\star=\pi\circ\mathfrak{S}_\text{ext}$ of $A_{\text{ext}\star}$ on $\mathcal{H}$. The vector $\vert 0\rangle$ 
might not be cyclic with respect to this representation and we define the Hilbert subspace $\mathcal{H}_\star = 
\overline{\pi_\star[A_{\text{ext}\star}]\vert 0\rangle}$. By the GNS-Theorem the cyclic representation 
$\bigl(\mathcal{H}_\star,\pi_\star,\vert 0 \rangle\bigr)$ is unitary equivalent to the GNS-representation 
of $(A_{\text{ext}\star},\Omega_\star)$. 

Using regularity of the state $\Omega$ we can define the 
(unbounded) hermitian field operators $\Phi(\varphi)\in\mathcal{L}(\mathcal{H})$, $\varphi\in V_\text{ext}$, 
as the generators of the Weyl operators $\pi\bigl(W(\varphi)\bigr)\in\mathcal{B}(\mathcal{H})$. They are related to the 
deformed field operators $\Phi_\star$ by $\Phi_\star(\varphi)= \Phi(S\varphi)\in\mathcal{L}(\mathcal{H})$,
which are the generators of $\pi_\star\bigl(W_\star(\varphi)\bigr) = \pi\bigl(W(S\varphi)\bigr)$.
Thus, $n$-point functions of the deformed field operators can be related to $n$-point functions of the
undeformed ones. Due to the quasifree assumption on the state $\Omega$ only the $2$-point function
is of interest and all higher $n$-point functions can be derived from it, in the undeformed and also deformed case.
We assume that $\vert 0 \rangle$ lies in the domain of $\Phi(\varphi)$, for all $\varphi\in V_\text{ext}$, and that
$\Phi(\psi)\vert 0\rangle$ lies in the domain of $\Phi(\varphi)$, for all $\varphi,\psi\in V_\text{ext}$.
 We define the deformed $2$-point function
\begin{flalign}
 \Omega_{\star2}(\varphi,\psi):= \langle 0\vert \Phi_\star(\varphi)\Phi_\star(\psi)\vert0\rangle=
\langle 0\vert \Phi(S\varphi)\Phi(S\psi)\vert0\rangle=\Omega_2(S\varphi,S\psi)~,
\end{flalign}
for all $\varphi,\psi\in V_\text{ext}$, where $\Omega_2$ is the undeformed $2$-point function.

In order to calculate the power spectrum we require the kernel of $\Omega_2$ ($\Omega_{\star2}$) in Fourier space.
Making use of the translation invariance of the state $\Omega$ we can define
\begin{flalign}
 \Omega_2(\varphi,\psi) = \int\limits_0^\infty dt t^3\int\limits_0^\infty d\tau \tau^3 \int\limits_{\mathbb{R}^3}
\frac{d^3k}{(2\pi)^3}~
\widetilde{\Omega}_2(t,\tau,k)\,\widetilde{\varphi}(t,k)\,\widetilde{\psi}(\tau,-k)~.
\end{flalign}
It follows that
\begin{flalign}
\label{eqn:def2point}
\Omega_{\star2}(\varphi,\psi) = 
 \Omega_{2}(S\varphi,S\psi) =
\int\limits_0^\infty dt t^3\int\limits_0^\infty d\tau \tau^3 \int\limits_{\mathbb{R}^3}
\frac{d^3k}{(2\pi)^3}~
\frac{\widetilde{\Omega}_2(t,\tau,k)}{\cosh(3\lambda k_1)}\,\widetilde{\varphi}(t,k)\,\widetilde{\psi}(\tau,-k)~.
\end{flalign}

The undeformed power spectrum is then given by
\begin{flalign}
 \mathcal{P}(t,k) := \widetilde{\Omega}_2(t,t,k)~,
\end{flalign}
and the deformed power spectrum reads
\begin{flalign}
\label{eqn:powerspectrum} 
\mathcal{P}_\star(t,k) = \frac{\mathcal{P}(t,k)}{\cosh(3\lambda k_1)}~.
\end{flalign}
The physical feature of $\mathcal{P}_\star$ compared to $\mathcal{P}$ is an exponential loss of power for large $\vert k_1\vert$.

\subsection{\label{ex:2dmodel}A spatially compact FRW toy-model}
The flow generated by the vector field $X_1=\partial_1$ in the model above is noncompact. We now investigate
if deformations along vector fields $X_1$ generating compact flows differ from the noncompact case.
A possible choice of toy-model would be the model of Section \ref{ex:frwqft} with
 $X_1=x^2\partial_3-x^3\partial_2$, i.e.~a rotation around the $x^1$-axis, and $X_2=t\partial_t$.
 
However, there is an even simpler model which we can use for our studies. Consider the manifold 
$\mathcal{M}= (0,\infty)\times S_1$, where $S_1$ is the one-dimensional circle, 
equipped with the metric $g=-dt\otimes dt+ t^2 d\phi\otimes d\phi$.
Here $t\in(0,\infty)$ denotes time and $\phi\in[0,2\pi)$ is the angle.  A proper homothetic Killing vector field is
$X_2=t\partial_t$ and $X_1=2\partial_\phi$ is a Killing vector field.\footnote{
The factor $2$ in the definition of $X_1$ is a convenient normalization.}  
The necessary condition $[X_1,X_2]=0$ is satisfied.
The calculation of the undeformed wave operator $P=\square_g-\xi\,\mathfrak{R}$ and the corresponding 
retarded and advanced Green's operators $\Delta_{\pm}$ is standard. We do not need to present the explicit results here.

We now investigate in detail the convergent version of $S=\sqrt{\cos(3\lambda\partial_\phi)^{-1}}$ and
$S^{-1}=\sqrt{\cos(3\lambda\partial_\phi)}$. 
For this we make use of the Fourier transformation on the spatial hypersurfaces $S_1$.
Since $S_1$ is compact, the momenta $n\in\mathbb{Z}$ are discrete. 
We define in Fourier space
\begin{subequations}
\begin{flalign}
\label{eqn:frw2S}\bigl(\widetilde{S}\widetilde{\varphi}\bigr)(t,n)&:= \sqrt{\cosh(3\lambda n)^{-1}}\,\widetilde{\varphi}(t,n)~,\\
\bigl(\widetilde{S}^{-1}\widetilde{\varphi}\bigr)(t,n)&:= \sqrt{\cosh(3\lambda n)}\,\widetilde{\varphi}(t,n)~.
\end{flalign}
\end{subequations}
It remains to discuss the domains of the maps $S$ and $S^{-1}$. Note that a function $\varphi$ on the circle $S_1$
is smooth if and only if its Fourier spectrum is of rapid decrease.
From this and (\ref{eqn:frw2S}) we infer $S: C^\infty_0(\mathcal{M})\to  C^\infty_0(\mathcal{M})$. 
However, $S^{-1}$ can not be defined on all of $C^\infty_0(\mathcal{M})$, provided we restrict the image of 
$S^{-1}$ to smooth functions. 
Since $S$ is injective, we find the isomorphism 
$S:C^\infty_0(\mathcal{M})\to C^\infty_\text{img}(\mathcal{M})\subset C^\infty_0(\mathcal{M})$.

Analogously to the situation where $X_1$ generates a noncompact flow we find that the symplectic isomorphism maps between
the deformed symplectic vector space $(V_\star,\omega_\star)$ and a nonstandard undeformed symplectic vector space
$(V_{\text{img}},\omega)$. In case of models with additional isometries $(V_{\text{img}},\omega)$ carries
only a representation of the unbroken subgroup. 
Furthermore, the map $S$ embeds $(V_\star,\omega_\star)$ into the undeformed  symplectic vector space $(V,\omega)$.
This is also analogous to the situation before, with the difference that we do not have to extend the 
symplectic vector space. All linear symplectic maps induce unique injective $\ast$-morphisms between
the corresponding Weyl algebras.

We find again that the deformed QFT looses observables, depending on the
value  of $\lambda$.
Since the functions in $C^\infty_{\text{img}}(\mathcal{M})$ have Fourier spectra which decrease
faster than $e^{-3\lambda \vert n\vert/2}$ for large $\vert n\vert$ we loose those observables
that are strongly localized in position space.

\subsection{On the physical inequivalence of the deformed and undeformed QFT}
In this subsection we collect arguments that our deformed QFT is physically inequivalent to 
standard commutative QFTs. 
This discussion is very important, since it has been shown \cite{Pohlmeyer}
that it is not possible in the framework of Wightman QFTs on
undeformed Minkowski spacetime to construct a nonlocal theory with
nonlocalities in the commutator function that fall off faster than
exponentially\footnote{
We thank the anonymous referee for bringing this to our attention.}.  
Indeed, Wightman QFTs with a faster than
exponentially vanishing nonlocality in the commutator function of
two fundamental fields are equivalent to local QFTs.  Therefore
we have to ensure that a similar result does not apply to our deformed models.

The first sign for an inequivalence of our deformed QFTs to standard commutative QFTs comes
from the nonlocal behavior of the deformed Green's operators $\Delta_{\star\pm}$ discussed above.
Coupling our deformed QFTs to external sources (or introducing perturbative interactions) then leads to 
nonlocal effects which in principle can be used to measure the deformation parameter $\lambda$.

The second sign comes from the form of the deformed power spectrum (\ref{eqn:powerspectrum}).
Commutative QFTs on FRW spacetimes typically have a power spectrum, which goes as 
$\mathcal{P}(t,k)\propto \vert k\vert^{n_s-4}$ for large $\vert k\vert$, where $n_s$ is
the spectral index. Our deformed power spectrum (\ref{eqn:powerspectrum}) shows, additionally to this power-law behavior,
an exponential drop-off for large $\vert k_1\vert$. Such a drop-off, and with this also the value of $\lambda$,
can in principle be measured.

The third sign comes from the nonlocality of the symplectic structure (\ref{eqn:nonlocalsymp}).
We consider analogously to \cite{Pohlmeyer} the correlation function
\begin{flalign}
\label{eqn:4pointcom} \langle 0\vert [\Phi_\star(\varphi_1),\Phi_\star(\varphi_2)]\Phi_\star(\varphi_3)\Phi_\star(\varphi_4)\vert 0\rangle~,
\end{flalign}
which, in our case, reduces due to the canonical commutation relations to
\begin{flalign}
 \nonumber \langle 0\vert [\Phi_\star(\varphi_1),\Phi_\star(\varphi_2)]\Phi_\star(\varphi_3)\Phi_\star(\varphi_4)\vert 0\rangle
&=i\,\omega_\star(\varphi_1,\varphi_2)\,\langle 0\vert \Phi_\star(\varphi_3)\Phi_\star(\varphi_4)\vert 0\rangle\\
&=
i\,\omega_\star(\varphi_1,\varphi_2)\,\Omega_{\star2}(\varphi_3,\varphi_4) ~.
\end{flalign}
Since, as explained above, the commutator function can be nonzero for functions $\varphi_1,\varphi_2$ with 
 spacelike separated support, (\ref{eqn:4pointcom}) is nonzero for these $\varphi_1,\varphi_2$ and for all
 $\varphi_3,\varphi_4$ such that $\Omega_{\star2}(\varphi_3,\varphi_4)\neq 0$.
Thus, our QFT is nonlocal.


\section{\label{sec:conc}Conclusions and outlook}
In this work we have studied a scalar QFT on certain NC self-similar symmetric spacetimes.
Our deformations are given by abelian Drinfel'd twists constructed from a Killing and a homothetic Killing vector field.
The deformed equation of motion and Green's operators have a particularly simple form, 
where the deformation factorizes. We have constructed the deformed  QFT in terms of a deformed 
$\ast$-algebra of field polynomials and have derived a $\ast$-algebra isomorphism connecting it to the
formal power series extension of the $\ast$-algebra of field polynomials of the corresponding undeformed QFT.
We have shown that the situation changes when we implement convergent deformations.
We have constructed the deformed Weyl algebra for  toy-models and have shown that it is $\ast$-isomorphic
to a reduced undeformed Weyl algebra, where strongly localized observables are excluded.

For future work there are two points of particular interest: Firstly, 
it would be interesting to study more general twists $\mathcal{F}$
constructed by Killing and homothetic Killing vector fields.
This would also result in a factorized deformed equation of motion operator $\widehat{P}_\star=Q_\mathcal{F}\circ (\square_g-\xi\,\mathfrak{R}) $, 
but now with $Q_\mathcal{F}$ more general than the cosine in (\ref{eqn:eomoperator}).
For models with spatial translation invariance we have the Killing vector fields $\partial_i$ and spatial Laplacian 
$\triangle = \delta^{ij}\partial_i\partial_j$.
Finding twists such that the operators $Q_\mathcal{F}$ only depend on $\triangle$ is of particular interest, since
this would lead to an isotropic modification of the QFT  at short length scales.
Secondly, the $D=2$ model of Section \ref{ex:2dmodel}, or the same model on $\mathcal{M}=(0,\infty)\times \mathbb{R}$,
is a very interesting toy-model to include interaction terms. The reason is that the deformed field operators represented 
in momentum space come with the factor $\sqrt{\cosh(3\lambda k)^{-1}}$, leading to an exponentially strong suppression
in the ultraviolet (UV). Thus, we expect that also the UV-properties of the interacting theory are strongly improved.
For the implementation of interactions the Yang-Feldman approach \cite{YFNC} may prove useful.


\section*{Acknowledgements}
I want to thank Claudio Dappiaggi, Eric Morfa-Morales, Thorsten Ohl, Christoph Uhlemann and Stefan Waldmann 
for discussions and comments on this work.
This research is supported by Deutsche Forschungsgemeinschaft through the Research 
Training Group GRK\,1147 \textit{Theoretical Astrophysics and Particle Physics}.

\appendix


\section{\label{app:twisted}On the twist deformation of the algebra of field polynomials}
In this appendix we study the twist approach to QFT, which is typically used for the Moyal-Weyl deformation
of a Minkowski QFT \cite{Zahn:2006wt,Balachandran:2007vx,Aschieri:2007sq}. We show that the twist deformation of the algebra
of field polynomials along homothetic Killing vector fields is possible if and only if the vector fields are Killing.
Note that in the approach \cite{Ohl:2009qe} to NC QFT no such restriction exist.

Let $\mathcal{A}_{(V,\omega)}$ be the formal power series extension of the $\ast$-algebra of field polynomials 
(see Definition \ref{def:fieldpoly}) of the commutative QFT. The basic idea of the twist approach to NC QFT is to replace 
the usual algebra product by a $\star$-product
\begin{flalign}
 a\star b=(\bar f^\alpha\triangleright a)\,(\bar f_\alpha \triangleright b)~,
\end{flalign}
for all $a,b\in\mathcal{A}_{(V,\omega)}$. We restrict ourselves to twists 
$\mathcal{F}^{-1}=\bar f^\alpha\otimes\bar f_\alpha$ generated by homothetic 
Killing vector fields $\mathfrak{H}$.
We assume the action $\triangleright$ of the twist on 
$\mathcal{A}_{(V,\omega)}$ to be the natural (geometric) action in order to interpret the deformation 
as a spacetime deformation. The geometric action of the Lie algebra $\mathfrak{H}$ is defined on the 
generators of $\mathcal{A}_{(V,\omega)}$ by 
\begin{subequations}
\begin{flalign}
 &v\triangleright 1:=0~,\\
 \label{eqn:geometricaction}&v\triangleright\Phi\bigl([\varphi]\bigr):= \Phi\bigl([\mathcal{L}_v(\varphi)]\bigr)~, 
\end{flalign}
\end{subequations}
for all $v\in\mathfrak{H}$ and $[\varphi]\in V$. The action is extended to $\mathcal{A}_{(V,\omega)}$ by 
$\mathbb{C}[[\lambda]]$-linearity and the Leibniz rule $v\triangleright (a\,b)=(v\triangleright a)\,b + a\,(v\triangleright b)$,
for all $a,b\in\mathcal{A}_{(V,\omega)}$ and $v\in\mathfrak{H}$.

It has to be checked if (\ref{eqn:geometricaction}) is well-defined. For this let
$[\varphi]=[\varphi^\prime]$, i.e.~$\varphi^\prime = \varphi + P(\psi)$, where $P=\square_g-\xi\,\mathfrak{R}$ is the equation
of motion operator and $\psi\in C^\infty_0(\mathcal{M},\mathbb{R})[[\lambda]]$. We find
\begin{flalign}
 \mathcal{L}_v(\varphi^\prime)= \mathcal{L}_v(\varphi)+\mathcal{L}_{v}\bigl(P(\psi)\bigr) = \mathcal{L}_v(\varphi)
 +P\bigl(\mathcal{L}_v(\psi) - c_v \psi\bigr)~,
\end{flalign}
where we have used that the scaling of the d'Alembert operator is $[\mathcal{L}_v,\square_g]=-c_v\,\square_g$.
Thus, the action is well-defined for all $v\in\mathfrak{H}$.

Next, we have to check if the action of $\mathfrak{H}$ is consistent with the commutation relations in $\mathcal{A}_{(V,\omega)}$.
We obtain the consistency condition 
(omitting the brackets $[\,\cdot\,]$ denoting equivalence classes)
\begin{multline}
\label{eqn:consistencycom}
 0=v\triangleright\bigl( i\, \omega(\varphi,\psi)\, 1\bigr)=v\triangleright[\Phi(\varphi),\Phi(\psi)]=
 [v\triangleright\Phi(\varphi),\Phi(\psi)]+[\Phi(\varphi),v\triangleright\Phi(\psi)]\\
 =[\Phi(\mathcal{L}_v(\varphi)),\Phi(\psi)]+[\Phi(\varphi),\Phi(\mathcal{L}_v(\psi))]=i\bigl(
\omega(\mathcal{L}_v(\varphi),\psi)+\omega(\varphi,\mathcal{L}_v(\psi))\bigr)\,1~,
\end{multline}
for all $\varphi,\psi\in V$ and $v\in\mathfrak{H}$.
Using the explicit form of $\omega$ we find
\begin{multline}
\label{eqn:leftright}
 \omega(\mathcal{L}_v(\varphi),\psi)=\int\limits_\mathcal{M}\mathcal{L}_v(\varphi)\,\Delta(\psi)\,\vol\stackrel{\text{PI}}{=}-
\int\limits_\mathcal{M}\varphi\,\mathcal{L}_v\bigl(\Delta(\psi)\,\vol\bigr)
\stackrel{\text{HKP}}{=}\\
-\int\limits_\mathcal{M}\varphi\,\left(\Delta\bigl(\mathcal{L}_v(\psi)\bigr) +c_v\,\left(\frac{D}{2}+1\right)\,
 \Delta(\psi) \right)\, \vol~=-\omega(\varphi,\mathcal{L}_v(\psi))-c_v\,\left(\frac{D}{2}+1\right)\,\omega(\varphi,\psi)~,
\end{multline}
for all $\varphi,\psi\in V$ and $v\in\mathfrak{H}$. In this derivation we have used integration by parts (PI)
and the homothetic Killing properties $\mathcal{L}_v(\vol)=\frac{ c_v D}{2}\vol$ and $[\mathcal{L}_v,\Delta]=c_v\Delta$
(HKP).
Putting (\ref{eqn:leftright}) into (\ref{eqn:consistencycom}) the consistency condition reads
\begin{flalign}
0= -c_v\,\left(\frac{D}{2}+1\right)\,\omega(\varphi,\psi)~,~\forall \varphi,\psi\in V~,
\end{flalign}
which implies $c_v=0$ due to the (weak) nondegeneracy of the symplectic structure $\omega$.
Thus, we can only represent the Lie subalgebra of Killing vector fields $\mathfrak{K}\subseteq\mathfrak{H}$ 
on $\mathcal{A}_{(V,\omega)}$, provided we assume a geometric action.

Let us briefly consider a general vector field $v\in\Xi$.
The two consistency conditions (\ref{eqn:geometricaction}) and (\ref{eqn:consistencycom}) required 
for $v$ to be implementable have the following meaning: (\ref{eqn:consistencycom}) states that
$v$ has to be an infinitesimal symplectic automorphism and (\ref{eqn:geometricaction})
means that the equation of motion operator has to transform as $[\mathcal{L}_v,P]=P\circ \mathcal{O}_v$, with some
operator $\mathcal{O}_v$ mapping compactly supported functions to compactly supported functions. 
These two conditions are of course not fulfilled for the most general vector field $v\in\Xi$.
Since the twisted QFT construction requires a representation of {\it all} vector fields entering the twist
on the algebra $\mathcal{A}_{(V,\omega)}$, the above argumentation shows that not all twists can be implemented. 
Our conjecture, which deserves a rigorous proof, is that a vector field is implementable if and only if it is Killing.

We now show for completeness that the twisted QFT construction is possible if the twist is Killing.
The Lie algebra representation of $\mathfrak{K}$ on $\mathcal{A}_{(V,\omega)}$ extends to a Hopf algebra
representation of the universal enveloping algebra $U\mathfrak{K}$, equipped with the natural coproduct, counit and antipode.
The formal deformation quantization of the Hopf algebra $U\mathfrak{K}$ and algebra $\mathcal{A}_{(V,\omega)}$ 
by a Killing twist $\mathcal{F}\in U\mathfrak{K}\otimes U\mathfrak{K}$ 
is then straightforward, extending the result for the Moyal-Weyl deformation of the Minkowski QFT
\cite{Zahn:2006wt,Balachandran:2007vx,Aschieri:2007sq}.


\end{document}